\documentclass[journal]{IEEEtran}
\usepackage{indentfirst}
\usepackage{multirow}
\usepackage{graphicx}
\usepackage{color}
\usepackage{xcolor}
\usepackage{colortbl}
\usepackage{listings}
\usepackage{enumerate}
\usepackage{subfigure}
\usepackage{amsmath,amsthm,amscd,amssymb,latexsym,upref,stmaryrd}
\usepackage[noend]{algpseudocode}
\usepackage{algorithmicx,algorithm}
\usepackage{multirow} 
\usepackage{xcolor}
\usepackage{amsfonts,epsfig}
\usepackage{CJK}
\usepackage{float}
\usepackage{cite}
\usepackage{subfigure}
\usepackage{lipsum}
\usepackage{multicol}
\usepackage{setspace}
\usepackage{bm}
\usepackage{stfloats}
\usepackage{booktabs}
\usepackage{threeparttable}

\newtheorem{Theorem}{Theorem}
\newtheorem{Corollary}{Corollary}%
\newtheorem{remark}{Remark}

\def\BibTeX{{\rm B\kern-.05em{\sc i\kern-.025em b}\kern-.08em
		T\kern-.1667em\lower.7ex\hbox{E}\kern-.125emX}}

\begin{document}
		\title{Hybrid Active and Passive Sensing for SLAM \\in Wireless Communication Systems}
		
	\author{Jie~Yang, \textit{Student Member, IEEE}, Chao-Kai~Wen, \textit{Senior Member, IEEE}, and Shi~Jin, \textit{Senior Member, IEEE}
	\thanks{Jie~Yang and Shi~Jin are with the National Mobile Communications Research Laboratory, Southeast University, Nanjing, China (e-mail: \{yangjie;jinshi\}@seu.edu.cn). Chao-Kai~Wen is with the Institute of Communications Engineering, National Sun Yat-sen University, Kaohsiung, 804, Taiwan (e-mail: chaokai.wen@mail.nsysu.edu.tw).}}
	
	\maketitle	
	
\begin{abstract}

Integrating sensing functions into future mobile equipment has become an important trend. Realizing different types of sensing and achieving mutual enhancement under the existing communication hardware architecture is a crucial challenge in realizing the deep integration of sensing and communication. In the 5G New Radio context, active sensing can be performed through uplink beam sweeping on the user equipment (UE) side to observe the surrounding environment. In addition, the UE can perform passive sensing through downlink channel estimation to measure the multipath component (MPC) information. This study is the first to develop a hybrid simultaneous localization and mapping (SLAM) mechanism that combines active and passive sensing, in which \emph{mutual enhancement} between the two sensing modes is realized in communication systems. Specifically, we first establish a common feature associated with the reflective surface to bridge active and passive sensing, thus enabling information fusion. Based on the common feature, we can attain physical anchor initialization through MPC with the assistance of active sensing. Then, we extend the classic probabilistic data association SLAM mechanism to achieve UE localization and continuously refine the physical anchor and target reflections through the subsequent passive sensing. Numerical results show that the proposed hybrid active and passive sensing-based SLAM mechanism can work successfully in tricky scenarios without any prior information on the floor plan, anchors, or agents. Moreover, the proposed algorithm demonstrates significant performance gains compared with active or passive sensing only mechanisms.
		
\end{abstract}


\section{Introduction} \label{sec:introduction}

Future communication networks are expected to simultaneously realize massive device connection, high-speed data transmission, and high-precision localization and sensing\cite{6G1}.
The exploitation of high frequencies and wide bandwidths, and the dense deployment of massive antenna arrays enable the implementation of integrated sensing and communication (ISAC) techniques \cite{ISAC,ISAC2,ISAC3,I2}.
The convergence paradigms of radar and communication have been widely investigated \cite{liufan,yonina,I5}.
Furthermore, the in-depth integration of radar and communication by sharing the same waveforms and hardware platforms has been preliminarily verified by using real systems \cite{radar2,I4,heath,AS,radar}.
Integrating sensing and communication functions into a single hardware platform and a joint signal processing framework helps reduce hardware costs, power consumption, and deployment complexity.
The cooperation and mutual assistance of the two functions can be conveniently achieved by utilizing the high-throughput, low-latency information sharing capability of wireless communication systems \cite{Location-aware}.
However, with the deepening of integration, many problems need to be solved urgently for ISAC designs.
In this study, we focus on the following key issues.
The first one is how to implement the sensing function at different stages of communication without extra cost.
The second one is how to fuse the sensing results of different communication stages and realize mutual assistance.
The third issue is how to achieve high sensing accuracy under strict practical constraints.


Current communication systems consist of beam sweeping, channel estimation, and data transmission stages, which are supported by 3rd Generation Partnership Project 5G new radio (NR) \cite{beam}. 
By using communication waveforms, different sensing types can be implemented in the ISAC system.
The methods that 
sense the radio environment by sending beamformed signals and analyzing the target reflections are called active sensing \cite{radar2,I4,heath,AS,radar}.
A personal mobile radar was proposed in \cite{radar2} by integrating massive antenna arrays into 5G mobile user devices, and a grid-based Bayesian mapping approach was utilized.
Moreover, a user-centric millimeter-wave indoor sensing system that operates at $28$ GHz with $400$ MHz bandwidth was realized in \cite{AS,radar}, which verified the capability of orthogonal frequency division multiplexing (OFDM)-based 5G NR waveform in sensing. 
The methods that capture the targets and surrounding environment through the received multipath signals sent by other terminals are called passive sensing \cite{UWB,SOO,yj3,yj4,loc2,CSLAM, CSLAM2,BP2,slam2,MVA,I1}.
Position-related information in multipath components of radio frequency (RF) signals was evaluated in \cite{UWB,SOO}.
Many studies have focused on fine-grained channel state information (CSI), such as high-resolution angle of arrival (AOA), angle of departure (AOD), time of arrival (TOA), and frequency of arrival (FOA)  \cite{yj3,yj4,loc2}.
The above existing research focused on a single sensing type, but two sensing types can be simultaneously achieved in a communication system.
However, fusing different types of sensing results is difficult.
The solution is to build a bridge for information fusion of different sensing types; then, the complementary and mutual assistance between different sensing types can be investigated, which has not been performed in existing literature.


In practical indoor scenarios, prior information about the environment is limited on the user side.
Sensing becomes increasingly challenging due to the complex multi-path propagation, severe path loss, and high probability of false alarm and missed detection. 
A feasible solution to solve tricky practical issues entails using multiple shots-based methods, which utilize measurements obtained in successive time slots.
A typical application is simultaneous localization and mapping (SLAM), which 
takes advantage of the unchanging nature of the location and state of radio features in the environment \cite{CSLAM, CSLAM2,slam2,BP2,MVA,I1 }.
Apart from radio-based SLAM, visual SLAM has achieved great success in estimating the trajectory of the camera while reconstructing the environment \cite{vslam}, where cameras can provide rich information of the environment that allows for robust and accurate place recognition. 
Besides, 
light detection and ranging (LiDAR)-based SLAM is one of the key technologies toward the success of autonomous driving and 3D imaging and measures distances by simply calculating the round-trip time of a laser pulse traveled to the target and back and offers centimeter-level resolution \cite{lidar}.
Although the resolution of radio is lower than that of camera or LiDAR, radio can cover a long detection range and is unaffected by weather and light conditions, which are the key challenges faced by visual and LiDAR SLAM. 
Accurate estimates of distance and angle are possible by exploiting antenna arrays and large bandwidth. Radio frequency signal designed for communication can help realize communication and sensing by completely reusing the communication hardware. Therefore, other devices in the communication network can participate in the collaboration-SLAM \cite{coslam}.

Radio features were abstracted into virtual anchors (VAs) in \cite{slam2,BP2}, which corresponded to different reflective surfaces and physical anchors.
The radio features were further abstracted into master virtual anchors (MVAs) in \cite{MVA}, which only corresponded to different reflective surfaces.
A sequential estimation of the states of a mobile agent and radio features was conducted in \cite{BP2,slam2,MVA}.
The estimation has low computational complexity and can cope with clutter, missed detections, and data associations.
However, the locations of physical anchors were assumed to be perfectly known to the mobile agent in \cite{BP2,slam2,MVA}.
Although the absolute location of agent can be obtained by formulating the relationship between agents, multipath measurements and the known floor plan\cite{I1}, this study focused on passive sensing only and did not further convert the active and passive sensing results into a general feature.
In reality, the anchors' states may be flexible (e.g., the anchors may be temporarily connected or disconnected).
Therefore, realizing SLAM in such challenging scenarios, in which the prior information of anchors and agents is unavailable, by the cooperation of active and passive sensing is worth exploring.

\begin{figure}
	\centering
	\vspace{-0.1cm}
	\includegraphics[scale=0.38]{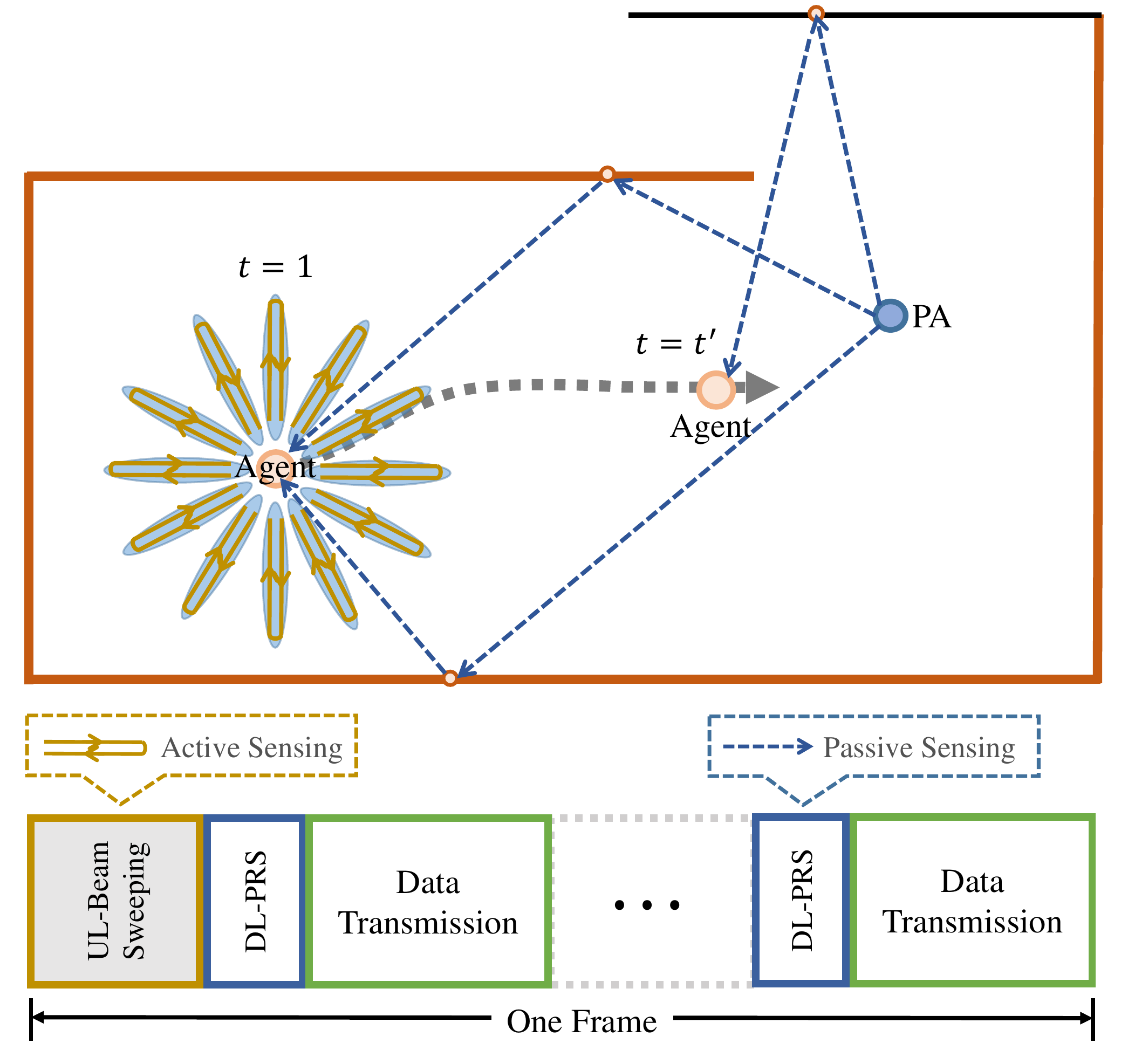}
	\caption{Scenario and frame structure of active and passive sensing in a wireless communication system. For illustration purposes, we only draw one PA. The grey dotted line is the trajectory of the agent, and we draw two positions of an agent at $t=1$ and $t=t'$.} \label{fig:frame}
	\vspace{-0.3cm}
\end{figure}

This study aims to develop a novel framework of SLAM in which active and passive sensing can cooperate with each other.
To our best knowledge, this study is the first to
realize mutual enhancement between active and passive sensing in communication systems.
Different from recent studies on radar-communication coexistence \cite{r9} and dual-functional radar-communication \cite{r10,r11,r13,r14}, we do not change the traditional transmitter of the 5G NR communication systems but integrate radio sensing into the current communication-only mobile network. The related communication-centric ISAC studies are summarized in detail in \cite{r17}.  
Many communication-centric ISAC studies make progresses in the aspects of beamforming design, power allocation \cite{r19}, and intercarrier interference utilization \cite{r20}.
Our main contributions are presented in detail as follows:
\begin{itemize}
	\item 
	We establish a general theoretical model to describe the result of active sensing.
	The theoretical model builds a bridge between active and passive sensing.	
	Therefore, the results of active and passive sensing can be transformed into the same type of features. 		
	With the proposed theoretical model, we can obtain the soft information of the features obtained by active sensing, including the mean and variance of the estimated features.	
	Thus, soft information fusion can be realized between active and passive sensing.

	\item 
	We realize the mutual assistance of active and passive sensing by extending the classic belief propagation (BP)-based SLAM method\cite{BP2,slam2,MVA}.
	Our development realizes physical anchor initialization with the assistance of active sensing and achieves feature refinement with the help of passive sensing.
	Therefore, compared with state-of-the-art approaches, the proposed hybrid active and passive sensing-based SLAM can work successfully in more realistic and challenging scenarios without any prior information on the floor plan, anchors, or agents.

\end{itemize}

The rest of this paper is organized as follows. Section \uppercase\expandafter{\romannumeral2} introduces the system model to bridge active and passive sensing. We derive the theoretical model for active sensing in Section \uppercase\expandafter{\romannumeral3}. 
The proposed hybrid active and passive sensing-based SLAM mechanism is described in Section \uppercase\expandafter{\romannumeral4}.
Our simulation results are presented in Section \uppercase\expandafter{\romannumeral5}, and we conclude the study in Section \uppercase\expandafter{\romannumeral6}.

\section{System Model and Problem Formulation }\label{s2}

We study the mutual assistance of active and passive sensing in wireless communication systems. 
As shown in Fig. \ref{fig:frame},  localization and mapping are performed on the agent side. 
Therefore, 
active and passive sensing can be realized by beam sweeping and downlink positioning reference signals (DL-PRS), respectively, in the communication process \cite{DL-PRS,3GPP}.
According to 5G NR beam management \cite{beam}, the agent can sweep the beams for initial access.
Different from conventional beam sweeping in communication systems, the agent meanwhile  listens to the reflected signal by the receive antenna array, thus enabling sensing the surrounding environment.
We assume that transmit and receive antenna arrays are placed separately on the agent to relax the self-interference \cite{radar,2sets} \footnote{A number of mobile phones have configured multiple sets of mmWave antenna modules on a mobile phone to overcome the hand blockage effect. Therefore, the functionality of separating  transmission and receiving can be achieved through different antenna modules.}. 
Therefore, during beam sweeping, the agent sends RF signals in beams then captures the reflective surfaces from the echo signals, which serve as active sensing.
For passive sensing, the agent captures environment features
through the received signals sent from the anchors.
The multipath component (MPC) information in passive sensing contains the state and location of anchors and reflective surfaces.
Motivated by the observation that the location and state of the reflective surfaces obtained by active and passive sensing are strongly related under the same environment, 
we perform theoretical modeling of the representation of the reflective surface and unify the results of active and passive sensing with those of the reflective surface.

\subsection{Geometric Model}\label{gm}
We consider an indoor scenario with $K$ static physical anchors (PAs) and a mobile agent.
The state of the mobile agent at time $t$ is denoted as $\mathbf{u}_{t} = [\mathbf{p}_{{\rm u},t},\mathbf{v}_{{\rm u},t}]$, where $\mathbf{p}_{{\rm u},t}=[x_{{\rm u},t},y_{{\rm u},t}]$ indicates the location and $\mathbf{v}_{{\rm u},t}=[\dot{x}_{{\rm u},t},\dot{y}_{{\rm u},t}]$ represents the velocity.
Let $\mathbf{p}^{(k)}_{{\rm pa},t}$
represent the location of the $k$-th PA at time $t$ for $k=1,\ldots,K$.
We denote VAs as mirror images of PAs on the reflective surfaces, as illustrated in Fig. \ref{fig:floorplan}.
Therefore, one single-bounce specular non-line-of-sight (NLOS) path corresponds to one VA.
Let $L_{t}$ denote the number of reflective surfaces.
The location of the $l$-th VA of the $k$-th PA at time $t$ is represented by $\mathbf{p}^{(k,l)}_{{\rm va},t}$, where $l=2,\ldots,L_{t}$.

We can set any point on the 2-dimensional plane as the reference point (RP). For ease of notation, we let RP be the start point of the mobile agent and let $\mathbf{p}_{\rm rp} = [x_{\rm rp},y_{\rm rp}]$ denote the location of RP.
Similar to \cite{MVA}, we introduce virtual reference points (VRPs), which are mirror images of RP on the reflective surfaces.
We denote the location of the $l$-th VRP as $\mathbf{p}^{(l)}_{{\rm vrp}, t}$, where $l=1,\ldots,L_{t}$.
The relationship among PA, VA, RP, and VRP is shown in Fig. \ref{fig:floorplan}. 
Let unit vector $\overrightarrow{\mathbf{n}}_l$ denote the normal direction of the $l$-th reflective surface. We have
\begin{equation}\label{mva}
\overrightarrow{\mathbf{n}}_l = \dfrac{ \mathbf{p}^{(l)}_{{\rm vrp}, t} - \mathbf{p}_{\rm rp} }{\|\mathbf{p}^{(l)}_{{\rm vrp}, t} - \mathbf{p}_{\rm rp}\|} = \dfrac{ \mathbf{p}^{(k,l)}_{{\rm va},t} - \mathbf{p}^{(k)}_{{\rm pa},t} }{\|\mathbf{p}^{(k,l)}_{{\rm va},t} - \mathbf{p}^{(k)}_{{\rm pa},t}\|}.
\end{equation}
where $\|\cdot\|$ signifies the L$_2$-norm.
VRP contains the information about the reflective surface, that is, the normal direction of the reflective surface and the distance between the reflective surface and RP.
The line of PA and VA is parallel to the line of RP and VRP. 
On the one hand, given the position of VRP and PA, VA is calculated as
\begin{equation}\label{mva-pa}
\mathbf{p}^{(k,l)}_{{\rm va},t} = \mathbf{p}^{(k)}_{{\rm pa},t} - 2\langle\mathbf{p}^{(k)}_{{\rm pa},t}-\dfrac{\mathbf{p}^{(l)}_{{\rm vrp}, t}+\mathbf{p}_{\rm rp}}{2} ,\overrightarrow{\mathbf{n}}_l\rangle \overrightarrow{\mathbf{n}}_l,
\end{equation}
where $ \langle \mathbf{a},  \mathbf{b} \rangle$ represents the inner product of vectors $\mathbf{a}$ and $\mathbf{b}$.
On the other hand, 
given the position of PA and VA, VRP is calculated as
\begin{equation}\label{pa-va}
\mathbf{p}^{(l)}_{{\rm vrp}, t} = \mathbf{p}_{\rm rp} - 
2\langle\mathbf{p}_{\rm rp}-\dfrac{\mathbf{p}^{(k)}_{{\rm pa},t}+\mathbf{p}^{(k,l)}_{{\rm va},t}}{2} ,\overrightarrow{\mathbf{n}}_l\rangle \overrightarrow{\mathbf{n}}_l,
\end{equation}
According to \eqref{mva-pa}, given the position of VRP and VA, PA is calculated as
\begin{equation}\label{mva-va}
\mathbf{p}^{(k)}_{{\rm pa},t} = \mathbf{p}^{(k,l)}_{{\rm va},t} - 2\langle\mathbf{p}^{(k,l)}_{{\rm va},t}-\dfrac{\mathbf{p}^{(l)}_{{\rm vrp}, t}+\mathbf{p}_{\rm rp}}{2} ,\overrightarrow{\mathbf{n}}_l\rangle \overrightarrow{\mathbf{n}}_l.
\end{equation}
Notably,  \eqref{mva-pa}-\eqref{mva-va} are equivalent.

\begin{figure}
	\centering\vspace{-0.2cm}
	\includegraphics[scale=0.52]{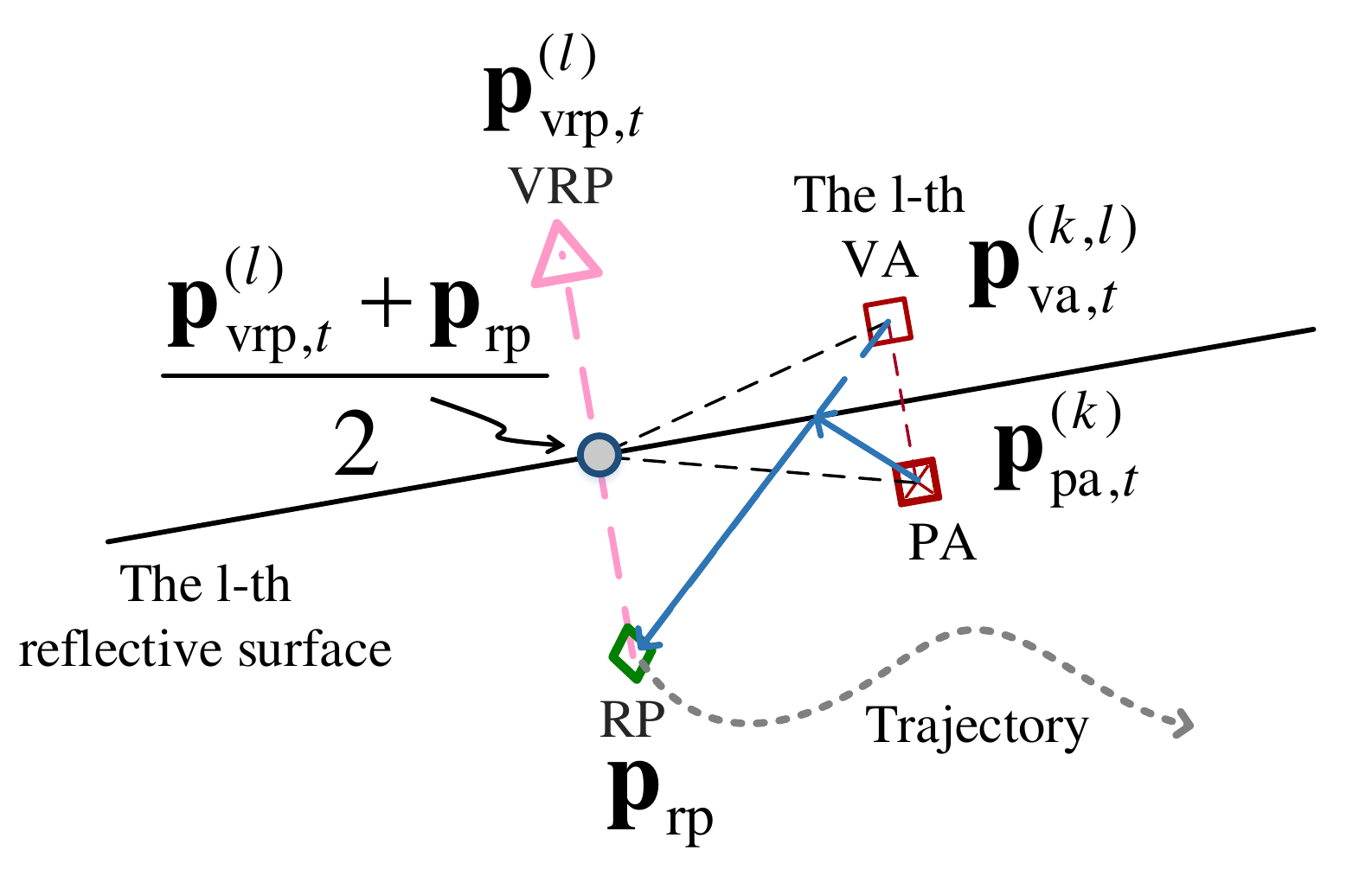}
	\caption{Relationship among PA, VA, VRP and RP. PA is denoted by a crossed box, VA is represented by a square, and VRP is denoted by a triangle.
		The trajectory of a mobile agent is depicted by a grey dotted line,	with RP being the start point.} \label{fig:floorplan}\vspace{-0.3cm}
\end{figure}

VRPs have three characteristics, namely, static location, changing presence, and reduced number.
The first characteristic is due to the assumption that the reflective surfaces are static.
The second characteristic arises from the fact that the agent encounters different reflective surfaces while moving.
As shown in Fig. \ref{fig:frame}, the black reflective surface is not detected at ${t=1}$ but is detected at ${t=t'}$; thus, the reflective surface is not present at ${t=1}$ but is present at ${t=t'}$.
The third characteristic is that VRP is determined by the reflective surface regardless of PAs.
Multiple VAs of different PAs may correspond to the same reflective surface.
Therefore, 
introducing VRPs can describe the same radio environment with less data than VAs.
Moreover,
one VA corresponds to one VRP, and different VAs of the same reflective surface share the same VRP.
Hence, introducing VRPs enables data fusion of multipaths from different PAs.

\subsection{Active Sensing}\label{AS}
We consider a mobile agent that has $N_{\rm b}$ uplink beams.
For a beam direction $\phi_i$, where $i = 1,\ldots, N_{\rm b}$, we assume that the agent transmits a sequence of beamformed sounding reference signals (SRSs) \cite{3GPP} on $N_{\rm s}$ active subcarriers, with ${x}_n$ denoting the pilot symbol at the $n$-th subcarrier in the OFDM symbol. 
In particular, the agent transmits SRS at the uplink beam direction $\phi_i$ while then observing the reflected signal at the same beam direction. Accordingly, the reflected signal reads \cite{AS,radar} 
\begin{equation}\label{reflection}
		{y}_{n,i} = \mathbf{a}^{\rm H}_{\rm RX}(\phi_i)\mathbf{H}_{n}\mathbf{a}_{\rm TX}(\phi_i){x}_{n} + {v},
\end{equation}
where $\mathbf{a}_{\rm RX}(\cdot)$ and $\mathbf{a}_{\rm TX}(\cdot)$ are steering vectors and $v$ is the additive Gaussian noise. We have 
\begin{equation}\label{channel}
	\mathbf{H}_{n} = \sum\limits_{m=1}^{M} \mathbf{a}_{\rm RX}(\phi_m)\Gamma_n(d_m)\mathbf{a}^{\rm H}_{\rm TX}(\phi_m),
\end{equation}
and 
\begin{equation}\label{gamma}
    \Gamma_n(d_m) = \frac{b_{n,m}}{d_m^2}e^{-j2\pi (n-\frac{N_{\rm s}}{2}) \Delta f \frac{2d_m}{c}}, \ {\rm with}\ \lvert b_{n,m} \rvert^2 = \frac{\lambda_n^2 \varepsilon_m}{(4\pi)^3},
\end{equation}	
where $M$ is the number of paths,
$d_m$ is the distance from the mobile agent to the reflection point on the reflective surface (RSP),
$\Delta f$ is the subcarrier spacing, and
$\lambda_n$ is the wavelength of the $n$-th subcarrier.
$\varepsilon_m$ is the radar cross section (RCS),
which we model by using $\varepsilon_m = \gamma \cos^{2\eta} \psi_m$, where $\psi_m$ is the angle between the incident wave and the normal direction of the reflective surface, in accordance with \cite{rcs} (Sec. 9.7.3).
In addition,
$\gamma$ and $\eta$ are model coefficients, and $c$ is the speed of light.
Therefore,
in the beam sweeping stage, we can obtain $d_m$ through range-angle processing methods \cite{radar}, where $m = 1,\ldots, M$.

For active sensing, the geometry relationship between the agent and the $m$-th RSP can be expressed as 
\begin{equation}\label{rsp}
	\mathbf{p}_{{\rm rsp},m} = \mathbf{p}_{{\rm u},t} + d_m[ \cos\phi_m,\sin\phi_m ]. 
\end{equation}
With more than two RSPs, we can determine the corresponding reflective surface.
Given that one VRP characterizes one reflective surface,
we take one VRP ($\mathbf{p}_{\rm vrp} = [x_{\rm vrp}, y_{\rm vrp}]$) and two RSPs ($\mathbf{p}_{{\rm rsp},1} = [x_{1}, y_{1}]$ and $\mathbf{p}_{{\rm rsp},2} = [x_{2}, y_{2}]$) as an example.
The relationships between VRP and RSPs are given by
\begin{equation}\label{vertical}
    (x_{\rm vrp}-x_{\rm rp})(x_{2}-x_{1})+(y_{\rm vrp}-y_{\rm rp})(y_{2}-y_{1}) = 0,
\end{equation}
and
\begin{equation}\label{collineation}
    \frac{\frac{x_{\rm vrp}+x_{\rm rp}}{2}-x_{1}}{\frac{y_{\rm vrp}+y_{\rm rp}}{2}-y_{1}} =\frac{x_{2}-x_{1}}{y_{2}-y_{1}}.
\end{equation}	 
As shown in Fig. \ref{fig:vrp}, \eqref{vertical} is established by the vertical relationship between vectors $\mathbf{p}_{{\rm rsp},1}-\mathbf{p}_{{\rm rsp},2}$ and $\mathbf{p}_{\rm vrp}-\mathbf{p}_{\rm rp}$,
and \eqref{collineation} is established according to the collinearity of three points, namely, $\mathbf{p}_{{\rm rsp},1}$, $\mathbf{p}_{{\rm rsp},2}$, and $(\mathbf{p}_{\rm vrp}+\mathbf{p}_{\rm rp})/2$.
Therefore, 
with active sensing to obtain two RSPs, we can derive VRP with  \eqref{rsp}-\eqref{collineation}. 

\begin{figure}
	\centering
	\includegraphics[scale=0.45]{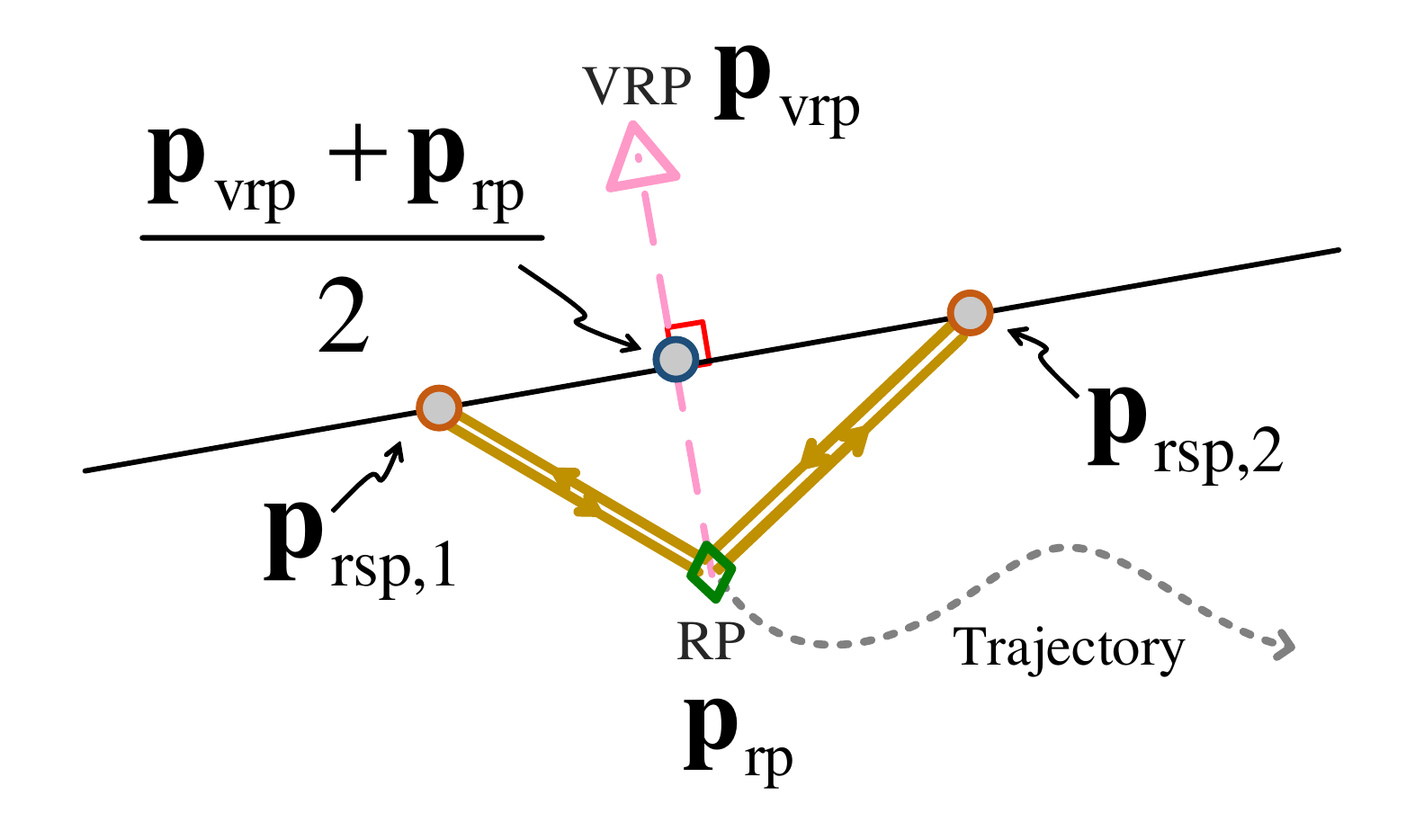}
	\caption{Relationship between RSPs and  VRP. Two RSPs and one VRP are illustrated.} \label{fig:vrp}\vspace{-0.3cm}
\end{figure}
	
\subsection{Passive Sensing}
In the 5G NR context, the $k$-th PA can be configured with a set of DL-PRS that are periodically transmitted to the mobile agent for positioning \cite{3GPP} for $k = 1, \ldots,K$.
The signal is transmitted over a multipath
channel with ${L}_{t}^{(k)}$ paths \cite{multipath}, but we only consider the line-of-sight (LOS) and single-bounce specular NLOS paths.
The AOAs, AODs, and TOAs of ${L}_{t}^{(k)}$ paths can be extracted from the received signal by advanced channel parameter extraction algorithms \cite{hy,nomp}.
For ease of expression, we consider only TOAs. The system model can be easily extended to solve cases with AOAs.
Let $\tau_{t,l}^{(k)}$ denote the TOA of the $l$-th path from the $k$-th PA at time slot $t$.
The geometry relationship between the agent and the $k$-th PA is given by \footnote{In this study, we assume that the agent and the physical anchor are synchronized, and the proposed mechanism can be expanded to the scenario where the clock bias is unknown. For the case where the clock bias between the agent and the physical anchor is a  constant, it can be viewed as an unknown attribute of the agent, and can be estimated with time. The detailed process is explained and realized in our previous work \cite{yj5}.  }
\begin{equation}\label{TOA}
c\tau_{t,1}^{(k)} ={\big\lVert  \mathbf{p}_{{\rm u},t}-\mathbf{p}^{(k)}_{{\rm pa},t} \big\lVert}.
\end{equation} 
The geometry relationship between the agent and the $l$-th VA is given by
\begin{equation}\label{TOA1}
c\tau_{t,l}^{(k)} ={\big\lVert  \mathbf{p}_{{\rm u},t}-\mathbf{p}^{(k,l)}_{{\rm va},t} \big\lVert},
\end{equation} 
where $c$ is the speed of light.
Let $\mathcal{M}_{t}^{(k)}$ represent a set of obtained TOA measurement indexes of the $k$-th PA at time $t$ and let
$\hat{\tau}_{t,l}^{(k)}$ denote the extracted TOA, where $l = 1,\ldots,|\mathcal{M}_{n}^{(m)}|$.
Notably, $|\cdot|$ denotes the number of elements in the set or the amplitude of a complex value.
We can then define the stacked measurement vectors
$\mathbf{z}_{t}^{(k)}=[\hat{\tau}_{t,1}^{(k)},\cdots,\hat{\tau}^{(k)}_{t,|\mathcal{M}_{t}^{(k)}|}]$ and
$\mathbf{z}_{t}=[\mathbf{z}_{t}^{(1)},\ldots,\mathbf{z}_{t}^{(K)}]$.
Through the accumulation of $T$ time slots, the agent can obtain a sequence of measurements $\mathbf{z}_{1:T}=[\mathbf{z}_{1},\ldots,\mathbf{z}_{T}]$.
The goal of passive sensing is to obtain the positions of PAs and VAs and the agent trajectory by measurements $\mathbf{z}_{1:T}$.
Therefore, with passive sensing to obtain PAs and VAs, we can derive VRP according to \eqref{pa-va}.
Hence, VRP acts as a bridge between active and passive sensing.

\subsection{Problem Formulation}
We consider a scenario where a mobile agent enters an unfamiliar indoor environment, and the positions of PAs and reflective surfaces are unknown to the agent.
Given that GPS is blocked indoors, the trajectory of the agent is also unknown.
The start point of the agent is viewed as RP.
Then, the agent gradually establishes the environment geometry relative to RP.

The agent performs beam sweeping at the start point for initial access and active sensing.
As shown in Fig. \ref{fig:frame}, the agent can obtain the information of the surrounding environment (the walls in red are the environment sensed by active sensing).
However, active sensing faces two challenges.
First, its sensing accuracy is affected by distance, while, a general model that describes the accuracy of active sensing is currently lacking. Therefore, we propose an uncertainty model for active sensing  (Section \ref{UM}).
Second, active sensing is inaccurate for distant targets and cannot handle obscured targets, such as the black wall in Fig. \ref{fig:frame}. Therefore, we propose a mechanism in which passive sensing can play an important role in refining the active sensing result (Section \ref{A&P}).

After the agent is connected to the PAs, passive sensing starts to work through a sequence of measurements $\mathbf{z}_{1:T}$.
With the VRPs obtained by active sensing, multi-paths of the same PA can be associated to obtain the initial PA position according to
\eqref{mva-va} (Section \ref{initial}).
When the agent moves along the trajectory, it has the chance to detect new reflective surfaces.
As shown in Fig. \ref{fig:frame}, the black wall is detected by passive sensing when $t=t'$.
Moreover, the estimates of VRPs and PAs 
become increasingly accurate over time (Section \ref{refine}).

With the help of active sensing, the locations of PAs do not need to be known.
With the assistance of passive sensing, 
the VRPs that correspond to new reflective surfaces can be accurately detected without being very close to the reflective surfaces, and the entire SLAM result is improved, as explained in Section \ref{result}.

\section{Uncertainty Model for Active Sensing }\label{UM}	

In this section, we establish a theoretical model that describes the accuracy of active sensing.
First, we model distance uncertainty. Second, we convert distance uncertainty to VRP uncertainty.  

\subsection{Distance Uncertainty}
According to Section \ref{AS}, the received reflection signal is given by \eqref{reflection}.
Considering that ${x}_{n}$ is the known pilot sequence and $\mathbf{a}^{\rm H}_{\rm RX}(\phi_m)\mathbf{a}_{\rm RX}(\phi_m)=1$ and $\mathbf{a}^{\rm H}_{\rm TX}(\phi_m)\mathbf{a}_{\rm TX}(\phi_m)=1$, we obtain
	\begin{equation}\label{ren}
	{r}_{n,m} = \Gamma_n(d_m) + {\tilde{v}},
	\end{equation}
	where ${\tilde{v}}$ follows a Gaussian distribution with zero mean and variance $\sigma^2$. 
	For $N_{\rm s}$ subcarriers, we
	denote $\mathbf{r}_{m}=[{r}_{1,m},\ldots,{r}_{N_{\rm s},m}]^{\rm T}$. Then, we extract the distance parameter $d_m$ from the received signal $\mathbf{r}_{m}$.
	
	According to \cite{KAY}, the information inequality for the variance of any unbiased estimate $\hat{d}_m$ reads  
    \begin{equation}\label{crlb}
    {\rm var}\{\hat{d}_m\} \geq {\rm F}^{-1}(d_m),
    \end{equation}
    where ${\rm F}(d_m)$ denotes the Fisher information matrix of $d_m$.
    ${\rm F}^{-1}(d_m)$ is also called the Cram\'{e}r Rao lower bound (CRLB).
    According to \cite{KAY}, ${\rm F}^{-1}(d_m)$ is defined by
    \begin{equation}\label{fim}
    {\rm F}(d_m) = \frac{2}{\sigma^2} \mathcal{R}\left\{
    \dfrac{\partial \textbf{r}_{m}^{\text{H}} }{\partial d_m}  \dfrac{\partial \textbf{r}_{m} }{\partial d_m} 
    \right\}.
    \end{equation}
    After some tedious calculation, we obtain the following theorem.    
    \begin{Theorem}\label{T1}
    	Let ${\rm SNR}_{m} = \dfrac{ \lambda^2 \varepsilon_m }{(4\pi)^3 d_m^4 \sigma^2}$ and $B = N_{\rm s}\Delta f$ denote the bandwidth. Then, the information inequality for the distance variance is given by 
    	\begin{equation}\label{variance}
    	{\rm var}\{\hat{d}_m\} \geq  \bigg( \dfrac{8\pi^2B^2N_{\rm s}}{3c^2} {\rm SNR}_{m} \bigg)^{-1}.
    	\end{equation}   	 
    \end{Theorem}
    \begin{proof}
    	Please refer to Appendix \ref{A}.
    \end{proof}
    \begin{remark}	
    The lower bound $ \big( \frac{8\pi^2B^2N_{\rm s}}{3c^2} {\rm SNR}_{m} \big)^{-1}$ of ${\rm var}\{\hat{d}_m\}$ in \eqref{variance} is defined as \textbf{distance uncertainty}, which increases with distance $d_m$ and noise variance $\sigma^2$ and decreases with bandwidth $B$, number of sub-carriers $N_{\rm s}$, wave length $\lambda$, and RCS $\varepsilon_m$.      
    The advanced parameter extraction technique can achieve CRLB. Therefore, the extracted distance $\hat{d}_m$ can be viewed as a Gaussian variable with mean $d_m$ and variance $\sigma_m^2 = \big( \frac{8\pi^2B^2N_{\rm s}}{3c^2} {\rm SNR}_{m} \big)^{-1}$.
    \end{remark}

    \subsection{VRP Uncertainty}
    VRP is calculated through the combination of points on the reflective surface according to \eqref{vertical} and \eqref{collineation}.   
    According to Theorem \ref{T1}, 
    the distance from $\mathbf{p}_{{\rm rsp},1}$ and $\mathbf{p}_{{\rm rsp},2}$ to RP follows a Gaussian distribution with mean $d_1$ and $d_2$ and variance $\sigma_1^2$ and $\sigma_2^2$, respectively.
    Therefore, for $\mathbf{p}_{{\rm rsp},1} = [x_{1}, y_{1}]$ and $\mathbf{p}_{{\rm rsp},2} = [x_{2}, y_{2}]$, we have  
    \begin{align}\label{wp1-pdf}
    &x_{1} \!\sim\! \mathcal{N}(d_1\cos\phi_1,\sigma_1^2\cos^2\phi_1), \ y_{1} \!\sim\! \mathcal{N}(d_1\sin\phi_1,\sigma_1^2\sin^2\phi_1),\nonumber \\
   & x_{2} \!\sim\! \mathcal{N}(d_2\cos\phi_2,\sigma_2^2\cos^2\phi_2), \ y_{2} \!\sim\! \mathcal{N}(d_2\sin\phi_2,\sigma_2^2\sin^2\phi_2),
    \end{align}
    where $\mathcal{N}(\mu,\sigma^2)$ denotes the Gaussian distribution with mean $\mu$ and variance $\sigma^2$.
    For representation simplicity, we denote     
    \begin{align}\label{wp1-mv}
    &\mu_{x1}=d_1\cos\phi_1, \  \ \sigma^2_{x1}=\sigma_1^2\cos^2\phi_1,  \nonumber \\ &\mu_{y1}=d_1\sin\phi_1, \ \ \sigma^2_{y1}=\sigma_1^2\sin^2\phi_1, \nonumber \\
    &\mu_{x2}=d_2\cos\phi_2, \ \ \sigma^2_{x2}=\sigma_2^2\cos^2\phi_2, \nonumber \\ &\mu_{y2}=d_2\sin\phi_2, \ \ \sigma^2_{y2}=\sigma_2^2\sin^2\phi_2.   
    \end{align}
    After some simplification on \eqref{vertical} and \eqref{collineation}, we have \eqref{x} and \eqref{y} at top of next page.
    \begin{figure*}
    \begin{equation}\label{x}
    x_{\rm vrp} = \dfrac{x_{\rm rp}(x_{2}-x_{1})^2-(x_{\rm rp}-2x_{1})(y_{2}-y_{1})^2+2(y_{\rm rp}-y_{1})(x_{2}-x_{1})(y_{2}-y_{1})}{(x_{2}-x_{1})^2+(y_{2}-y_{1})^2},
    \end{equation}
    \end{figure*}
    \begin{figure*}
    \begin{equation}\label{y}
    y_{\rm vrp} = \dfrac{y_{\rm rp}(y_2-y_1)^2-(y_{\rm rp}-2y_1)(x_2-x_1)^2+2(x_{\rm rp}-x_1)(x_2-x_1)(y_2-y_1)}{(x_2-x_1)^2+(y_2-y_1)^2}.
    \end{equation}
    \hrulefill
    \end{figure*}
    Notably, $x_{\rm rp}$ and $y_{\rm rp}$ are constants, and $x_1$, $y_1$, $x_2$, and $y_2$ follow Gaussian distributions.
    Directly obtaining the distributions of $x_{\rm vrp}$ and $y_{\rm vrp}$ from \eqref{x} and \eqref{y} is difficult due to the non-linear calculations \cite{ratio}.
    We linearize \eqref{x} and \eqref{y} via first-order Taylor approximation.  
    The approximate distributions of $x_{\rm vrp}$ and $y_{\rm vrp}$ are given in the following theorem.
    
    \begin{Theorem}\label{T2}                
        Let $\mathbf{w}=[x_1-\mu_{x1},x_2-\mu_{x2},y_1-\mu_{y1},y_2-\mu_{y2}]$,         then, the linear approximation of $\mathbf{p}_{\rm vrp}$ with $\mathbf{w}$ is given by 
        \begin{equation}\label{pappro}
        \mathbf{p}_{\rm vrp} \approx \mathbf{p}_0 + \mathbf{w}\mathbf{Q},
        \end{equation}
        where
        $\mathbf{p}_0 = \big[\frac{a_0}{b_0}, \frac{c_0}{b_0}\big]$ and
        $\mathbf{Q}=[\mathbf{q}_x,\mathbf{q}_y]$ with
	    \begin{equation}\label{QX}
	    \mathbf{q}_x
	    \!\!=\!\!\bigg[\!\dfrac{a_1\!b_0\!-\!a_0b_1}{b_0^2}\!,\dfrac{a_2b_0\!-\!a_0b_2}{b_0^2}\!,\dfrac{a_3b_0\!-\!a_0b_3}{b_0^2}\!,\dfrac{a_4b_0\!-\!a_0b_4}{b_0^2}\!\bigg]^{\rm\!\! T}\!\!\!, \!\!\end{equation}
	    and
	    \begin{equation}\label{QY}
	    \mathbf{q}_y\!\!=\!\!\bigg[\!\dfrac{c_1b_0\!-\!c_0b_1}{b_0^2}\!, \dfrac{c_2b_0\!-\!c_0b_2}{b_0^2}\!, \dfrac{c_3b_0\!-\!c_0b_3}{b_0^2}\!,\dfrac{c_4b_0\!-\!c_0b_4}{b_0^2}\!\bigg]^{\rm\!\! T}\!\!\!,
	    \end{equation}
        where $(\cdot)^{\rm T}$ represents the transpose and the constants $a_0,\ldots, a_4$, $b_0,\ldots, b_4$, and $c_0,\ldots, c_4$ are given in Appendix \ref{B}.
    	\end{Theorem}
     \begin{proof}
    	Please refer to Appendix \ref{B}.
    \end{proof}
    \begin{remark}
    	Given that the elements of $\mathbf{w}$ follow Gaussian distributions and $\mathbf{p}_{\rm vrp}$ has an approximately linear relationship with $\mathbf{w}$, we assume that $\mathbf{p}_{\rm vrp}$ follows a Gaussian distribution with mean $\mathbf{p}_0$ and covariance matrix $\mathbf{Q}^{\rm T}{\rm cov}\{\mathbf{w}\}\mathbf{Q}$, where ${\rm cov}\{\mathbf{w}\} = {\rm diag}\{\sigma^2_{x1},\sigma^2_{x2},\sigma^2_{y1},\sigma^2_{y2}\}$.
    \end{remark}

    When $M$ RSPs are considered, we obtain $M-1$ solutions of VRP. 
    According to Theorem \ref{T2}, by combining RSP $m$ and RSP $1$, the $m$-th VRP solution is obtained as
    \begin{equation}\label{pm}
    	\mathbf{p}_{\rm vrp}^{1m} \approx \mathbf{p}_0^{1m} + \mathbf{w}^{1m}\mathbf{Q}^{1m},
    \end{equation}	
    where \eqref{pm} is obtained by replacing RSP $2$ in \eqref{pappro} with RSP $m$.
    The distribution of $\mathbf{p}_{\rm vrp}$ is given in the following corollary.	
    	\begin{Corollary}\label{L2}	
    	By combining the $M-1$ solutions, the mean of VRP is obtained as
    	\begin{equation}\label{M-mean} 
    	\frac{1}{M-1} {\sum_{m=2}^{M}\mathbf{p}_0^{1m}},
    	\end{equation}	
    	and the covariance matrix is given as
    	\begin{equation}\label{M-cov}   \frac{1}{(M-1)^2}{\sum_{m=2}^{M} (\mathbf{Q}^{1m})^{\rm T} {\rm cov}\{\mathbf{w}^{1m}\}}\mathbf{Q}^{1m}.
    	\end{equation} 	
    \end{Corollary}

    \begin{remark}	 
    The covariance matrix given in \eqref{M-cov} is 
    defined as \textbf{ VRP uncertainty}. 
    In accordance with \eqref{M-cov}, the more RSPs we have, the smaller the VRP uncertainty. 
    The RCS characterizes the reflection condition at different points on the reflecting surface.
    We can infer from \eqref{M-cov} that one or two points with poor reflection conditions will not affect the overall result as long as we have sufficient number of RSPs. 
    The estimated VRP can be assumed to follow the Gaussian distribution, with mean $\frac{1}{M-1} {\sum_{m=2}^{M}\mathbf{p}_0^{1m}}$ and covariance matrix $\frac{1}{(M-1)^2}{\sum_{m=2}^{M} (\mathbf{Q}^{1m})^{\rm T}} {\rm cov}\{\mathbf{w}^{1m}\}\mathbf{Q}^{1m}$.
    	
    \end{remark}

\section{Hybrid Active and Passive Sensing for SLAM}\label{A&P}	

\begin{table*}
	\centering
	\caption{Notations of important variables}\label{NOTATIONS}
	\renewcommand{\arraystretch}{1.5}
	\fontsize{7.6}{7.6}\selectfont
	\begin{threeparttable}	
		\begin{tabular}{c l c l}			
			\toprule		
			Notation & Definition & Notation & Definition \\
			\hline
			$\mathbf{p}^{(k)}_{t,l}$ & $l=1$: location of the $k$-th PA at time $t$                  & $\mathbf{z}_{t,l}^{(k)}$ & measurement vector corresponding to the $l$-th VA \\  
			& $l>1$: location of the $l$-th VRP at time $t$ &                          & of the $k$-th PA at time $t$                      \\ \hline
			$\mathbf{u}_{t}$         & state of the mobile agent at time $t$                         & $\mathbf{a}_t^{(k)}$     & feature-oriented data association vector          \\ \hline
			$L_{t}^{(k)}$            & number of features corresponding to the $k$-th PA             & $\mathbf{b}_t^{(k)}$     & measurement-oriented data association vector      \\ \hline
			 $r^{(k)}_{t,l}$          & binary variable indicates the existence of feature & ${\mathbf{v}}^{(k)}_{t,l}$ & state of the $l$-th feature at time $t$                                   \\ \hline
			$\mathcal{F}_{t}^{(k)}$ & set of measurement indexes of false alarms  & $\mathbf{c}_{t}$           & number-of-measurements vector at time $t$                                                 \\ \hline
			$\mathcal{K}_{t-\mathbb{J}}^{(k-1+\mathbb{I})}$     & set of legacy features from the previous time & $\mathcal{M}_{t}^{(k)}$ & set of obtained MPC measurement indexes of the $k$-th PA                                                            \\ \hline
			$\mathcal{D}_{t}^{(k)}$ & set of legacy feature indexes which generate measurement & $\bar{\mathcal{D}}_{t}^{(k)}$ & $\bar{\mathcal{D}}_{t}^{(k)}= \mathcal{K}_{t-\mathbb{J}}^{(k-1+\mathbb{I})} \backslash {\mathcal{D}}_{t}^{(k)}$                     \\ \hline
			$\mathcal{N}_{t}^{(k)}$ & set of measurement indexes which originate from new features & $\bar{\mathcal{N}}_{t}^{(k)}$ & $\bar{\mathcal{N}}_{t}^{(k)}= \mathcal{M}_{t}^{(k)} \backslash {\mathcal{N}}_{t}^{(k)}$                       \\ \hline
			$\tilde{\ast}$ & legacy $\ast$, $\ast = \{   \mathbf{p}^{(k)}_{t,l},  r^{(k)}_{t,l}, {\mathbf{v}}^{(k)}_{t,l}\}$  & $ \breve{\star}$ & new $\star$, $\star = \{   \mathbf{p}^{(k)}_{t,l},  r^{(k)}_{t,l}, {\mathbf{v}}^{(k)}_{t,l}\}$  \\				
			\bottomrule			
		\end{tabular}
	\end{threeparttable}
\end{table*}

In this section, we propose a hybrid active and passive sensing-based BP SLAM mechanism, as shown in Fig. \ref{fig:as}.
Different from the classic passive sensing-based BP SLAM \cite{MVA}, the proposed mechanism has two more modules: (i) PA initialization and (ii) VRP and PA refinement. 
Therefore, compared with mechanisms based only on passive or active sensing, the proposed mechanism has two advantages:
(i) the proposed mechanism does not require any prior information about PAs because the active sensing result is used in the PA initialization, and
(ii) the active sensing result is gradually refined by passive sensing.
The important variables are summarized in Table \ref{NOTATIONS}.

\begin{figure*}
	\centering
	\includegraphics[scale=0.6]{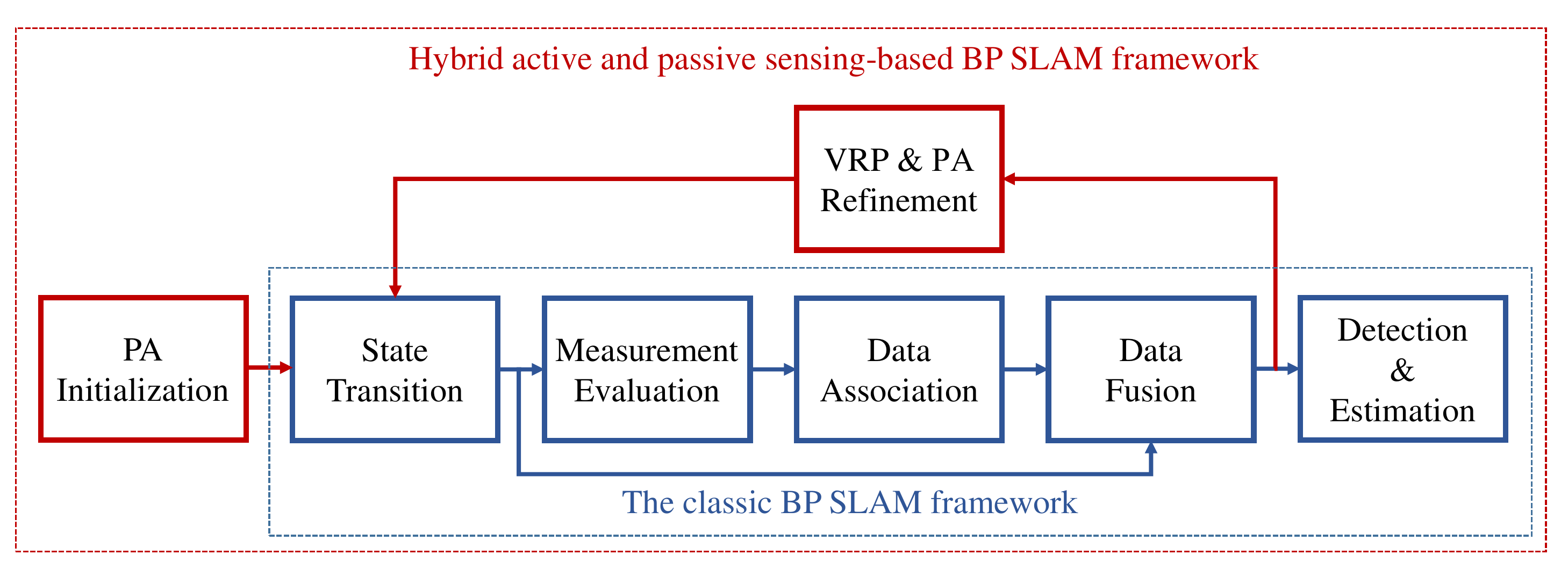}
	\caption{Framework of the proposed hybrid active and passive sensing-based BP SLAM.} \label{fig:as}
\end{figure*}

\subsection{Theoretical Foundation of VRP-based BP SLAM}
In this section, we introduce the theoretical foundation of VRP-based BP SLAM, as shown in the blue modules in Fig. \ref{fig:as}.
Let PAs and VRPs denote the features of the radio environment,
where $\mathbf{p}^{(k)}_{{\rm pa},t} = \mathbf{p}^{(k)}_{t,1} = [x^{(k)}_{t,1},y^{(k)}_{t,1}]$ and 
$\mathbf{p}^{(k,l)}_{{\rm vrp},t} = \mathbf{p}^{(k)}_{t,l} = [x^{(k)}_{t,l},y^{(k)}_{t,l}]$ for $l\geqslant 2$.
Therefore, the location of the $l$-th feature is denoted by $\mathbf{p}^{(k)}_{t,l}$.
Given the data association uncertainty, 
a measurement can originate from a \emph{legacy feature} or a \emph{new feature} or may not originate from any feature (i.e., a \emph{false alarm}).
Different from VA-based BP SLAM \cite{slam2,BP2,yj5}, in the VRP-based BP SLAM,
the measurements and features are associated sequentially across PAs.
Therefore,
a legacy feature for the $k$-th PA means that the feature already exists for the  $(k-1)$-th PA at time $t$ when $k\neq 1$ or for the $K$-th PA at time $t-1$ when $k=1$.
Let $k-1+\mathbb{I}$ denote the index of the legacy feature at time slot $t-\mathbb{J}$,
where 
\begin{equation}
(\mathbb{I},\mathbb{J})  = \left\{ 
\begin{array}{ll}
& \!\!\!\!\!\!\!(K,1), \ \ \text{if} \  k=1,  \\ 
&\!\!\!\!\!\!\! (0,\ 0), \ \  \text{otherwise}. 
\end{array}
\right.
\end{equation}
Let $\tilde{\mathbf{p}}^{(k)}_{t,l}$ denote the location of a legacy feature that exists at the current time slot and $\mathcal{D}_{t}^{(k)}$ denote the set of legacy feature indexes that generate a measurement of the current PA at the current time.
By contrast, a new feature means the feature does not exist for the $(k-1+\mathbb{I})$-th PA at time $t-\mathbb{J}$.
Let $\breve{\mathbf{p}}^{(k)}_{t,l}$ denote the location of a new feature and $\mathcal{N}_{t}^{(k)}$ denote a set of measurement indexes originating from new features.
Moreover, let $\mathcal{F}_{t}^{(k)}$ denote a set of measurement indexes of false alarms. Therefore, we classify the measurement indexes in $\mathcal{M}_{t}^{(k)}$ into three subsets according to their origins and obtain
$
|\mathcal{M}_{t}^{(k)}| = |\mathcal{D}_{t}^{(k)}|+|\mathcal{N}_{t}^{(k)}|+|\mathcal{F}_{t}^{(k)}|
$.

For data association, we define the following data association vectors according to \cite{BP1}.
First, let $\mathcal{K}_{t-\mathbb{J}}^{(k-1+\mathbb{I})}$ represent a set of legacy feature indexes.
The $|\mathcal{K}_{t-\mathbb{J}}^{(k-1+\mathbb{I})}|$-dimensional feature-oriented vector is  $\mathbf{a}_t^{(k)}=\big[{a}_{t,1}^{(k)},\ldots,{a}_{t,|\mathcal{K}_{t-\mathbb{J}}^{(k-1+\mathbb{I})}|}^{(k)}\big]$, the element of which is given by 
\begin{equation}\label{DA1}
{a}_{t,i}^{(k)}\!=\!\left\{
\begin{array}{ll}
\!\!\!j \in \{ 1,\ldots, |\mathcal{M}_{t}^{(k)}|\}, &\!\!\text{legacy feature $i$ generates} \\
&\!\!\text{measurement $j$ at time $t$,} \\
\!\!\!0, &\!\!\text{legacy feature $i$ does not}\\
&\!\!\text{generate any measurement},
\end{array} \right.
\end{equation}where $i = 1,\ldots, |\mathcal{K}_{t-\mathbb{J}}^{(k-1+\mathbb{I})}|$.
We also define the stacked vector $\mathbf{a}_t=[\mathbf{a}_{t}^{(1)},\ldots,\mathbf{a}_{t}^{(K)}]$.
Second, the $|\mathcal{M}_{t}^{(k)}|$-dimensional measurement-oriented vector is $\mathbf{b}_t^{(k)}=\big[{b}_{t,1}^{(k)},\ldots,{b}_{t,|\mathcal{M}_{t}^{(k)}|}^{(k)}\big]$, and we obtain 
\begin{equation}\label{DA2}
{b}_{t,j}^{(k)}\!=\!\left\{
\begin{array}{ll}
\!\!\! i \!\in \!\{ 1,\ldots, |\mathcal{K}_{t-\mathbb{J}}^{\!(k-1\!+\mathbb{I})\!}| \}, & \!\!\text{measurement $j$ is generated} \\ 
& \!\!\text{by legacy feature $i$ at time $t$,} \\ 
\!\!\! 0, &\!\!\text{measurement $j$ is not} \\
&\!\!\text{generated by legacy feature,} 
\end{array} \right.
\end{equation}
where $j = 1,\ldots, |\mathcal{M}_{t}^{(k)}|$.
We also define the stacked vector $\mathbf{b}_t=[\mathbf{b}_{t}^{(1)},\ldots,\mathbf{b}_{t}^{(K)}]$.
Vectors $\mathbf{a}_t$ and $\mathbf{b}_t$, which are equivalent because one can be determined from the other,
can ensure the scalability properties of the BP algorithm. 
A constraint exists such that each measurement originates from a maximum of one feature or one false alarm, and one feature can generate at most one measurement each time.
The exclusion-enforcing
function used to ensure the constraint is defined as
\begin{equation}\label{indicator}
\Psi(\mathbf{a}_{t}^{(k)},\mathbf{b}_{t}^{(k)})= \prod\limits_{i=1}^{|\mathcal{M}_{t}^{(k)}|}\prod\limits_{j=1}^{|\mathcal{K}_{t-\mathbb{J}}^{(k-1+\mathbb{I})}|}\Psi({a}_{t,i}^{(k)},{b}_{t,j}^{(k)}),
\end{equation}
where 
\begin{equation}
\Psi({a}_{t,i}^{(k)},{b}_{t,j}^{(k)}) = \left\{
\begin{array}{ll}
0,  &  {a}_{t,i}^{(k)}=j, {b}_{t,j}^{(k)} \neq i\  \text{or} \  {b}_{t,j}^{(k)}=i, {a}_{t,i}^{(k)} \neq j,\\
1, &  \text{otherwise}.
\end{array} \right.
\end{equation}

Then, the joint posterior probability density function (PDF) of the state of the agent and features and the data association vectors conditioned on measurements for all time slots up to $T$ is defined as
\begin{multline}\label{jointP}
f(\mathbf{u}_{1:T},{\mathbf{v}}_{1:T},\mathbf{a}_{1:T},\mathbf{b}_{1:T}|\mathbf{z}_{1:T}) \\
= \prod\limits_{t=1}^{T} \prod\limits_{k=1}^{K} f(\mathbf{u}_{t},{\mathbf{v}}_{t}^{(k)},\mathbf{a}_{t}^{(k)},\mathbf{b}_{t}^{(k)}|\mathbf{z}_{t}^{(k)}),
\end{multline}
which can be computed based on Bayes' theorem follows:
\begin{multline}\label{jointP1}
f(\mathbf{u}_{1:T},{\mathbf{v}}_{1:T},\mathbf{a}_{1:T},\mathbf{b}_{1:T}|\mathbf{z}_{1:T})\\
\!\propto \! \prod\limits_{t=1}^{T}\!\prod\limits_{k=1}^{K}\!\underbrace {f(\mathbf{u}_{t},\tilde{{\mathbf{v}}}_{t}^{(k)}|\mathbf{u}_{t-1},{\mathbf{v}}_{t-\mathbb{J}}^{(k-1+\mathbb{I})})}_{(a)}\\
\times
\underbrace {f(\mathbf{z}_{t}^{(k)}|\mathbf{u}_{t},{\mathbf{v}}_{t}^{(k)},\mathbf{a}_{t}^{(k)},\mathbf{b}_{t}^{(k)})}_{(b)}
\\
\times\underbrace {f(\mathbf{a}_{t}^{(k)},\mathbf{b}_{t}^{(k)},\mathbf{c}^{(k)}_{t},\breve{{\mathbf{v}}}_{t}^{(k)}|\tilde{{\mathbf{v}}}_{t}^{(k)},\mathbf{u}_{t})}_{(c)},
\end{multline}
where $\tilde{{\mathbf{v}}}_{t}^{(k)}$ and $\breve{{\mathbf{v}}}_{t}^{(k)}$ denote the state of legacy and new features, respectively, and 
${\mathbf{v}}_{t}^{(k)}=[\tilde{{\mathbf{v}}}_{t}^{(k)},\breve{{\mathbf{v}}}_{t}^{(k)}]$.
The elements in $\tilde{{\mathbf{v}}}_{t}^{(k)}$ and $\breve{{\mathbf{v}}}_{t}^{(k)}$ are denoted by $\tilde{\mathbf{v}}^{(k)}_{t,l} = [\tilde{\mathbf{p}}^{(k)}_{t,l},\tilde{r}^{(k)}_{t,l}]$ for $l=1,\ldots, |\mathcal{D}_{t}^{(k)}|$ and $\breve{\mathbf{v}}^{(k)}_{t,l} = [\breve{\mathbf{p}}^{(k)}_{t,l},\breve{r}^{(k)}_{t,l}]$ for $l=1,\ldots,|\mathcal{N}_{t}^{(k)}|$, respectively.
Binary variables $\tilde{r}^{(k)}_{t,l}\in\{0,1\}$ and
$\breve{r}^{(k)}_{t,l}\in\{0,1\}$ indicate the existence of the $(k,l)$-th feature at time $t$, that is, the feature exists at time $t$ if and only if $\tilde{r}^{(k)}_{t,l}=1$ or $\breve{r}^{(k)}_{t,l}=1$. 
The number-of-measurements vector at time $t$ is $\mathbf{c}_t=[{c}_{t}^{(1)},\ldots,{c}_{t}^{(K)}]$, where ${c}_{t}^{(k)}=|\mathcal{M}_{t}^{(k)}|$.
We have $\mathbf{a}_{t}$ implies $\mathbf{b}_{t}$ and $\mathbf{z}_t$ implies $\mathbf{c}_{t}$. 
Notably, (a), (b), and (c) of \eqref{jointP1} correspond to the state transition, measurement evaluation, and data association phases, respectively.
The data fusion phase corresponds to the entire process of \eqref{jointP1} (Fig. \ref{fig:as}).

Then, a minimum mean squared error (MMSE) estimator for the agent's state $\mathbf{u}_{T}$ at time slot $T$ is given as  
\begin{equation}\label{mmse1}
\hat{\mathbf{u}}_{T} = \int \mathbf{u}_{T} f(\mathbf{u}_{T}|\mathbf{z}_{1:T})  \text{d} \mathbf{u}_{T}, 
\end{equation}
where $f(\mathbf{u}_{T}|\mathbf{z}_{1:T})= \int_{\mathbf{x}}  f(\mathbf{u}_{1:T},{\mathbf{v}}_{1:T},\mathbf{a}_{1:T},\mathbf{b}_{1:T}|\mathbf{z}_{1:T}) \text{d}\mathbf{x} $  is a marginal posterior PDF, and $\mathbf{x} = [\mathbf{u}_{1:T\!-\!1},\!{\mathbf{v}}_{1:T},\!\mathbf{a}_{1:T},\!\mathbf{b}_{1:T}]$.
The posterior existence probability $p({r}_{T,l}^{(k)}=1|\mathbf{z}_{1:T})$ is given as
\begin{equation}\label{marginal}
p({r}_{T,l}^{(k)}\!=\!1|\mathbf{z}_{1:T}) =\int \!\!\!f(\mathbf{{p}}_{T,l}^{(k)},\!{r}_{T,l}^{(k)}\!=\!1|\mathbf{z}_{1:T}) \text{d} \mathbf{p}_{T,l}^{(k)} ,
\end{equation} 
where $f(\mathbf{{p}}_{T,l}^{(k)},\!{r}_{T,l}^{(k)}\!=\!1|\mathbf{z}_{1:T})$ is a marginal posterior PDF in \eqref{jointP1}.
On the basis of Bayes' theorem, we obtain 
\begin{equation}\label{marginal2}
f(\mathbf{p}_{T,l}^{(k)}|{r}_{T,l}^{(k)}\!=\!1,\mathbf{z}_{1:T})  \!=\! \frac{ f(\mathbf{p}_{T,l}^{(k)},{r}_{T,l}^{(k)}\!=\!1|\mathbf{z}_{1:T})}{p({r}_{T,l}^{(k)}=1|\mathbf{z}_{1:T})}. 
\end{equation}
The MMSE estimator for the feature $\mathbf{p}_{T,l}^{(k)}$ can be obtained as 
\begin{equation}\label{mmse6}
\hat{\mathbf{p}}_{T,l}^{(k)} \! =\!\! \int \!\! \mathbf{p}_{T,l}^{(k)} f(\mathbf{p}_{T,l}^{(k)}|{r}_{T,l}^{(k)}=1,\mathbf{z}_{1:T})  \text{d} \mathbf{p}_{T,l}^{(k)}. 
\end{equation}
The detection phase in Fig. \ref{fig:as} is \eqref{marginal}, and the estimation phase is \eqref{mmse1} and \eqref{mmse6}.

\begin{figure*}
	\centering
	\includegraphics[scale=0.205]{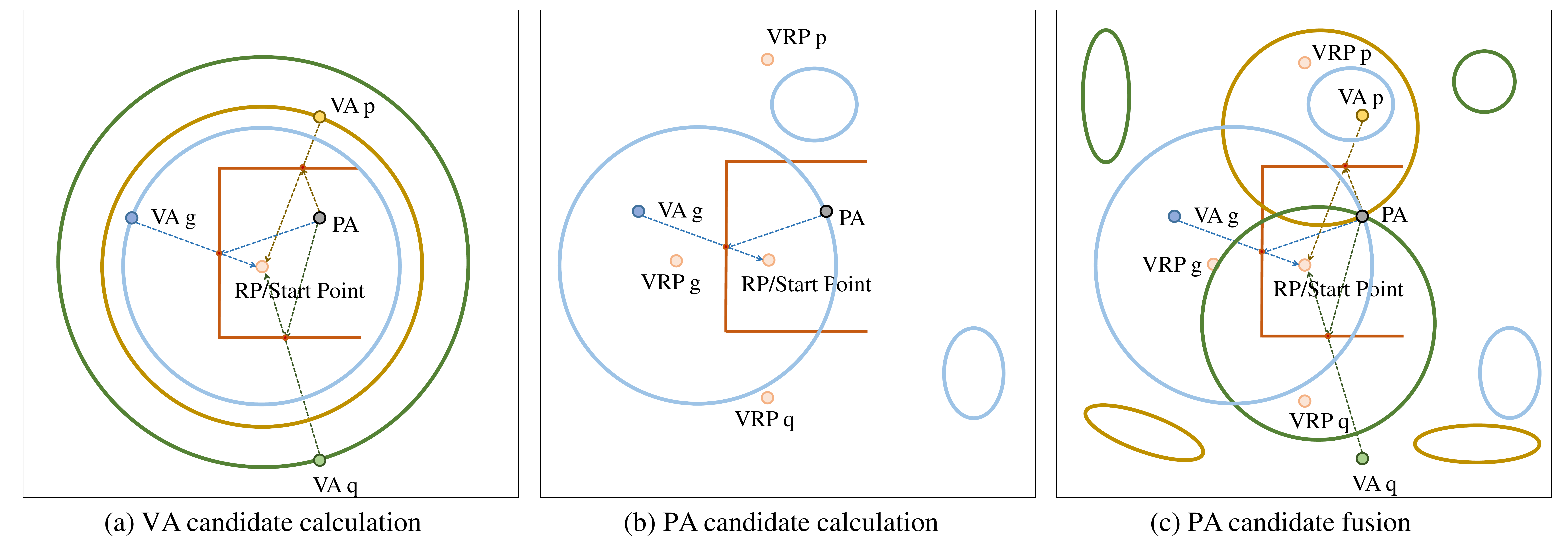}
	\caption{Illustration of PA initialization, where one PA, three VRPs, and three VAs are depicted as an example.} \label{fig:painitial}	
\end{figure*}
\subsection{Hybrid Active and Passive Sensing-based SLAM Mechanism}
In this section, we explain each phase (Fig. \ref{fig:as}) of the proposed hybrid active and passive sensing-based SLAM mechanism in detail.

\subsubsection{PA Initialization}\label{initial}

At the start point of the trajectory, 
the agent performs active sensing to obtain $N_{\rm vrp}$ VRPs. Then, the agent performs passive sensing to obtain $|\mathcal{M}_{1}^{(k)}|$ measurements corresponding to PA $k$ for $k=1,\cdots,K$.
PA initialization undergoes three steps:
(i) VA candidate calculation, (ii) PA candidate calculation, (iii) PA candidate fusion.

The process of PA initialization with one PA is shown in Fig. \ref{fig:painitial} as an example.
VA candidate calculation is shown in Fig. \ref{fig:painitial} (a). Here, we obtain one VA candidate by \eqref{TOA} for one measurement. 
The VA candidate $\mathbf{p}^{(k,l)}_{{\rm va},t}$ is distributed on the circle centered on RP with radius $c\tau_{t,l}^{(k)}$.
The measurements have three kinds of origins, that is, PA, VA, and false alarm.
The VA candidate can be a VA, PA, or false alarm, where PA can be excluded by removing the smallest TOA measurement.
According to \eqref{mva-va}, we can localize PA if we know VRP and VA.
PA candidate calculation is shown in Fig. \ref{fig:painitial} (b) for one VA.
We combine the VA candidate and $N_{\rm vrp}$ VRPs one by one in accordance with \eqref{mva-va} because the data association among VRPs and VAs is uncertain.
After transformation by \eqref{mva-va}, the circle of a VA candidate is transformed into a circle when combined with the right VRP, and the circle of a VA candidate is transformed into an irregular shape when combined with the wrong VRPs.
We use an ellipse to represent the irregular shape;
thus, the blue circle and ellipses in Fig. \ref{fig:painitial} (b) represent the distribution of PA candidates corresponding to VA g.
Meanwhile, the PA candidate fusion is shown in Fig. \ref{fig:painitial} (c), where three VRPs and three VAs are considered.
Given that one VA candidate generates three PA candidates, by associating the PA candidates generated by three or more VAs, we can obtain the true position of PA.
The process is repeated $K$ times in parallel to obtain the initialization of $K$ PAs.

\subsubsection{State Transition}

The agent and legacy feature states are assumed to independently evolve according to Markovian state dynamics given by 
\begin{multline}\label{mark}
f(\mathbf{u}_{t},\tilde{{\mathbf{v}}}_{t}^{(k)}|\mathbf{u}_{t-1},{\mathbf{v}}^{(k-1+\mathbb{I})}_{t-\mathbb{J}})\\
=f(\mathbf{u}_{t}|\mathbf{u}_{t-1})\prod\limits_{k=1}^{K}\prod\limits_{l=1}^{|\mathcal{K}_{t-\mathbb{J}}^{(k-1+\mathbb{I})}|}f(\tilde{{\mathbf{v}}}^{(k)}_{t,l}|{\mathbf{v}}^{(k-1+\mathbb{I})}_{t-\mathbb{J},l}).
\end{multline}
The state transition function of agent $f(\mathbf{u}_{t}|\mathbf{u}_{t-1})$
is defined by a linear, near-constant-velocity motion model \cite{sem} given as $\mathbf{u}_{t}^{\rm T} = \mathbf{A} \mathbf{u}_{t-1}^{\rm T}+ \bm{\omega}_{t}$, 
where
\begin{equation}       
\mathbf{A}=\left(                
\begin{array}{cccc}   
1 & 0 & \Delta T & 0\\ 
0 & 1 & 0 & \Delta T\\  
0 & 0 & 1 & 0\\  
0 & 0 & 0 & 1\\  
\end{array}
\right),
\end{equation}
$\Delta T$ is the sampling period,
and $\bm{\omega}_{t}$ is the driving process, that follows an independently identically Gaussian distribution across $t$ with zero mean.
For the state transition function of feature $f(\tilde{{\mathbf{v}}}^{(k)}_{t,l}|{\mathbf{v}}_{t-\mathbb{J}}^{(k-1+\mathbb{I})})$,
if a feature does not exist at the previous time, then it cannot exist as a legacy feature at the current time. Therefore,
for ${\tilde{{r}}}_{t-\mathbb{J},l}^{(k-1+\mathbb{I})}=0$, we obtain
\begin{equation}\label{tt1}
\begin{array}{ll}
&f(\tilde{{\mathbf{p}}}^{(k)}_{t,l},{\tilde{{r}}}^{(k)}_{t,l}|\mathbf{{p}}_{t-\mathbb{J},l}^{(k-1+\mathbb{I})},0)=\left\{
\begin{array}{lcc}f_{\rm D}(\tilde{{\mathbf{p}}}^{(k)}_{t,l}), & & {\tilde{{r}}^{(k)}_{t,l}=0},\\
0, & & {\tilde{{r}}^{(k)}_{t,l}=1},
\end{array} \right.
\end{array} 
\end{equation}
where $f_{\rm D}(\cdot)$ is an arbitrary ``dummy" PDF, which is explained in detail in \cite{BP1}.
If a feature exists at the previous time, then the probability that it still exists at the current time is determined by the survival probability.
Therefore, for ${\tilde{{r}}}_{t-\mathbb{J},l}^{(k-1+\mathbb{I})}=1$, we obtain
\begin{equation}\label{tt2}
\begin{array}{ll}
&f(\tilde{{\mathbf{p}}}^{(k)}_{t,l},{\tilde{{r}}}^{(k)}_{t,l}|\mathbf{{p}}_{t-\mathbb{J},l}^{(k-1+\mathbb{I})},1)\\
&
=\left\{
\begin{array}{lc}
\left(1-P_{\rm s}(\mathbf{{p}}_{t-\mathbb{J},l}^{(k-1+\mathbb{I})})\right)f_{\rm D}(\tilde{{\mathbf{p}}}^{(k)}_{t,l}),  &  {\tilde{{r}}^{(k)}_{t,l}=0},\\
{P_{\rm s}(\mathbf{{p}}_{t-\mathbb{J},l}^{(k-1+\mathbb{I})})} f(\tilde{{\mathbf{p}}}^{(k)}_{t,l}|\mathbf{{p}}_{t-\mathbb{J},l}^{(k-1+\mathbb{I})})  &  {\tilde{{r}}^{(k)}_{t,l}=1},
\end{array} \right.
\end{array} 
\end{equation}
where ${P_{\rm s}(\cdot)}\in (0,1]$ represents the survival probability of a feature.

\subsubsection{Measurement Evaluation}

A measurement can originate from a legacy feature, new feature, or false alarm.
We define the likelihood function, that is, the PDF of measurements conditioned on the agent, features, and two data association vectors, as follows: 
\begin{multline}\label{mmp}
f(\mathbf{z}^{(k)}_{t}|\mathbf{u}_{t},{\mathbf{v}}_{t}^{(k)},\mathbf{a}_{t}^{(k)},\mathbf{b}_{t}^{(k)}) \\
=\!\!\!\! \prod\limits_{i\in\mathcal{D}_{t}^{(k)}}\!\!  \!\!\!f(\mathbf{z}^{(k)}_{t,{a}_{t,i}^{(k)}}|\mathbf{u}_{t},\tilde{{\mathbf{v}}}_{t}^{(k)})
\!\!\!\!\prod\limits_{ j\in\mathcal{N}_{t}^{(k)}}\!\!\!\!\!f(\mathbf{z}^{(k)}_{t,j}|\mathbf{u}_{t},\breve{{\mathbf{v}}}_{t}^{(k)}) 
\!\!\!\! \prod\limits_{q\in\mathcal{F}_{t}^{(k)}}\!\!\!\!\!f_{\rm false}(\mathbf{z}^{(k)}_{t,q}),
\end{multline}
where $\mathcal{D}_{t}^{(k)} \triangleq\left\{i \in\{1, \ldots, |\mathcal{K}_{t-\mathbb{J}}^{(k-1+\mathbb{I})}|\}: a_{t, i}^{(k)} \neq 0\right\}$.
The likelihood function is updated to 
\begin{multline}\label{mmmp2}
f(\mathbf{z}^{(k)}_{t}|\mathbf{u}_{t},{\mathbf{v}}_{t}^{(k)},\mathbf{a}_{t}^{(k)},\mathbf{c}_{t}^{(k)}) \\
\propto   \prod\limits_{i\in\mathcal{D}_{t}^{(k)}} \!\! \dfrac{f(\mathbf{z}^{(k)}_{t,{a}_{t,i}^{(k)}}|\mathbf{u}_{t},\tilde{{\mathbf{v}}}_{t}^{(k)})}{f_{\rm false}(\mathbf{z}^{(k)}_{t,{a}_{t,i}^{(k)}})}  \prod\limits_{j\in\mathcal{N}_{t}^{(k)}}  \dfrac{f(\mathbf{z}^{(k)}_{t,j}|\mathbf{u}_{t},\breve{{\mathbf{v}}}_{t}^{(k)})}{f_{\rm false}(\mathbf{z}^{(k)}_{t,j})},
\end{multline} 
where the number of false alarms and newly detected features follows a Poisson distribution with a mean of $\mu_{\rm false}^{(k)}$ and $\mu_{\rm new}^{(k)}$, respectively.
The distribution of each false alarm measurement is described by the PDF $f_{\rm false}(\cdot)$.

\begin{figure*}
	\centering
	\includegraphics[scale=0.56]{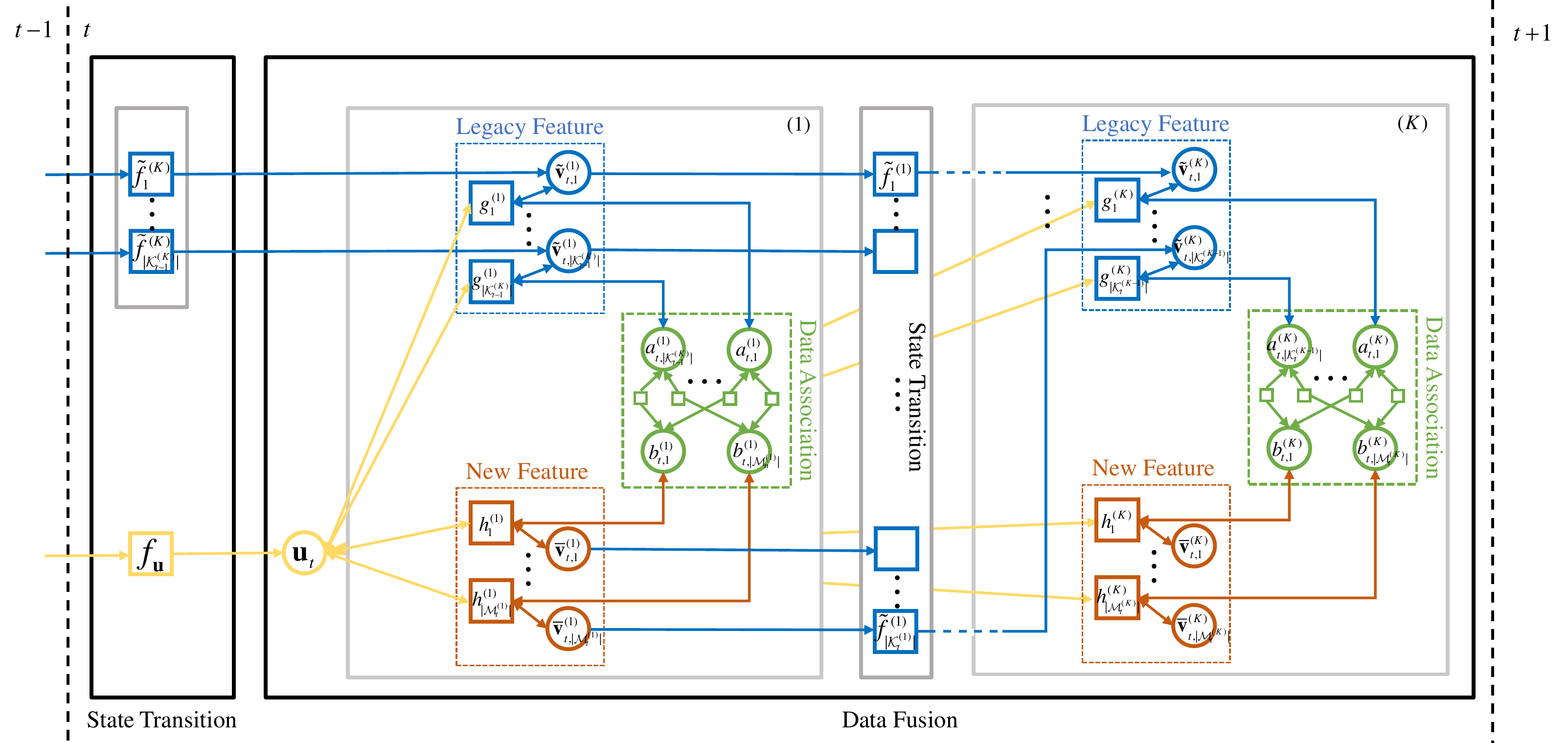}
	\caption{Factor graph of VRP-based BP SLAM.} \label{fig:fg}	
\end{figure*}

\subsubsection{Data Association}\label{DA}
The joint prior PDF of data association vectors,  number-of-measurements vector, and state of new features conditioned on the state of legacy features and the agent is  
\begin{multline}\label{da1}
f(\mathbf{a}_{t}^{(k)},\mathbf{b}_{t}^{(k)},\mathbf{c}_{t}^{(k)},\breve{{\mathbf{v}}}_{t}^{(k)}|\tilde{{\mathbf{v}}}_{t}^{(k)},\mathbf{u}_{t}) \\ 
\propto   \Psi(\mathbf{a}_{t}^{(k)},\mathbf{b}_{t}^{(k)})(\mu_{\rm new}^{(k)})^{|\mathcal{N}_{t}^{(k)}|}(\mu_{\rm false}^{(k)})^{-|\mathcal{N}_{t}^{(k)}|-|\mathcal{D}_{t}^{(k)}|}\\
\times  \!\!\prod\limits_{i\in\mathcal{D}_{t}^{(k)}}\!\!\!P_{\rm d}(\mathbf{u}_{t},\mathbf{{p}}_{t,{a}_{t,i}^{(k)}}^{(k)}) \!\!\!\!\prod\limits_{i'\in\bar{\mathcal{D}}_{t}^{(k)}}\!\!\!\!\left(1\!-\!P_{\rm d}(\mathbf{u}_{t},\mathbf{{p}}_{t,i'}^{(k)})\right)\!\!\!\!
\\
\times \prod\limits_{j\in\mathcal{N}_{t}^{(k)}}\!\! \!\!f_{\rm new}(\breve{{\mathbf{v}}}_{t,j}^{(k)}|\mathbf{u}_{t})\!\!\!\!\prod\limits_{j'\in\bar{\mathcal{N}}_{t}^{(k)}}\!\!\!\! f_D(\breve{{\mathbf{v}}}^{(k)}_{t,j'}),
\end{multline}
where $\bar{\mathcal{D}}_{t}^{(k)}= \mathcal{K}_{t-\mathbb{J}}^{(k-1+\mathbb{I})} \backslash {\mathcal{D}}_{t}^{(k)}$, $\bar{\mathcal{N}}_{t}^{(k)}= \mathcal{M}_{t}^{(k)} \backslash {\mathcal{N}}_{t}^{(k)}$, ``$\backslash$" represents the complement operator, $P_{\rm d}(\cdot) \!\!\in\!\! (0,1]$ is 
the probability that a feature is ``detected" in the sense that it generates a measurement, and $f_{\rm new}(\cdot)$ represents the PDF of the newly detected features.

\subsubsection{Data Fusion}
The joint posterior PDF of the agent, features, and data association vectors conditioned on measurements for all $T$ time slots is $f(\mathbf{u}_{1:T}, {\mathbf{v}}_{1:T},\mathbf{a}_{1:T},\mathbf{b}_{1:T}|\mathbf{z}_{1:T})$.
According to \eqref{jointP1}, the joint posterior PDF is the product of \eqref{mark}, \eqref{mmmp2}, and \eqref{da1}.
Given that the factorizations of \eqref{mmmp2} and \eqref{da1} 
are in the perspective of legacy and new features, we rewrite \eqref{mmmp2} and \eqref{da1} in a more concise form.
${\mathcal{D}}_{t}^{(k)}$ is the set of legacy feature indexes that generate measurements at time $t$. 
For $i \in {\mathcal{D}}_{t}^{(k)}$,
we have $\tilde{r}_{t,i} = 1$ and ${a}_{t,i}^{(k)}\neq 0$. 
On the contrary,
$\bar{\mathcal{D}}_{t}^{(k)}$ is the set of legacy feature indexes that do not generate measurements at time $t$. For $i \in \bar{\mathcal{D}}_{t}^{(k)}$,
we have $\tilde{r}_{t,i} = 1$ and ${a}_{t,i}^{(k)}= 0$.
We define a function $g(\mathbf{u}_{t},\tilde{\mathbf{v}}_{t,i}^{(k)},\mathbf{a}^{(k)}_{t,i};\mathbf{z}^{(k)}_{t,i})$. When $\tilde{r}_{t,i} = 1$, we have 
\begin{multline}\label{g1}	g(\mathbf{u}_{t},\tilde{\mathbf{v}}_{t,i}^{(k)},\mathbf{a}^{(k)}_{t,i};\mathbf{z}^{(k)}_{t,i})\\
=
\left\{
\begin{array}{ll}
\dfrac{f(\mathbf{z}^{(k)}_{t,{a}_{t,i}^{(k)}}|\mathbf{u}_{t},\tilde{{\mathbf{v}}}_{t}^{(k)})P_{\rm d}(\mathbf{u}_{t},\mathbf{{p}}_{t,{a}_{t,i}^{(k)}}^{(k)})}{\mu_{\rm false}^{(k)}f_{\rm false}(\mathbf{z}^{(k)}_{t,{a}_{t,i}^{(k)}})}, &   {a}_{t,i}^{(k)}\neq 0,\\
1-P_{\rm d}(\mathbf{u}_{t},\mathbf{{p}}_{t,i}^{(k)}), & {a}_{t,i}^{(k)}= 0.
\end{array} \right.
\end{multline}
When $\tilde{r}_{t,i}\!\! =\!\! 0$, we have
$
g(\!\mathbf{u}_{t},\!\tilde{\mathbf{v}}_{t,i}^{(k)}\!,\!\mathbf{a}^{(k)}_{t,i};\mathbf{z}^{(k)}_{t,i}\!)=1.
$
Moreover, ${\mathcal{N}}_{t}^{(k)}$ denotes the set of measurement indexes generated by new features, which means that for $j \in {\mathcal{N}}_{t}^{(k)}$, we have $\breve{r}_{t,j}\!\!=\!\!1$ and ${b}_{t,j}^{(k)}\!\!=\!\!0$. By contrast,
$\bar{\mathcal{N}}_{t}^{(k)}$ denotes the set of measurements that are not generated by new features, and we have $\breve{r}_{t,j}\!\!=\!\!0$ for $j \in \bar{\mathcal{N}}_{t}^{(k)}$.
We define a function $h(\mathbf{u}_{t},\breve{\mathbf{v}}_{t,j},\mathbf{b}^{(k)}_{t,j};\mathbf{z}^{(k)}_{t,j})$. When $\breve{r}_{t,j}=1$, we have 
\begin{multline}\label{h1}
h(\mathbf{u}_{t},\breve{\mathbf{v}}_{t,j},\mathbf{b}^{(k)}_{t,j};\mathbf{z}^{(k)}_{t,j})\\
=\left\{
\begin{array}{ll}
  0,  &     {b}_{t,j}^{(k)}\neq 0,\\
  \dfrac{\mu_{\rm new}^{(k)}f_{\rm new}(\breve{\mathbf{{v}}}_{t,j}^{(k)}| \mathbf{u}_{t})f(\mathbf{z}^{(k)}_{t,j}|\mathbf{u}_{t},\breve{\mathbf{v}}_{t}^{(k)})}{\mu_{\rm false}^{(k)}f_{\rm false}(\mathbf{z}^{(k)}_{t,j})}, &     {b}_{t,j}^{(k)}= 0.
\end{array} \right.
\end{multline}
When $\breve{r}_{t,j}\!\!=\!0$, we have  
$
h(\mathbf{u}_{t},\breve{\mathbf{v}}_{t,j},\mathbf{b}^{(k)}_{t,j};\mathbf{z}^{(k)}_{t,j}) = f_D(\breve{\mathbf{{v}}}_{t,j}).
$
The joint posterior PDF is given by \eqref{final},
\begin{figure*}
\begin{multline}\label{final}
	f(\mathbf{u}_{1:T},{\mathbf{v}}_{1:T},\mathbf{a}_{1:T},\mathbf{b}_{1:T}|\mathbf{z}_{1:T})
	\propto   \underbrace{\prod\limits_{t=1}^{T}f(\mathbf{u}_{t}|\mathbf{u}_{t-1})\prod\limits_{k=1}^{K}\prod\limits_{l=1}^{|\mathcal{K}_{t-\mathbb{J}}^{(k-1+\mathbb{I})}|}f(\tilde{{\mathbf{v}}}^{(k)}_{t,l}|{\mathbf{v}}^{(k-1+\mathbb{I})}_{t-\mathbb{J},l})}_{(a)}\\
 \times \prod\limits_{t=1}^{T}\left(\underbrace{\prod\limits_{k=1}^{K} 
		\Psi(\mathbf{a}_{t}^{(k)}\!,\!\mathbf{b}_{t}^{(k)}\!)}_{(c)}  \underbrace{\!\!\prod\limits_{i=1}^{|\mathcal{K}_{t-\mathbb{J}}^{(k-1+\mathbb{I})}|} \!\! g(\mathbf{u}_{t},\tilde{\mathbf{v}}_{t,i}^{(k)}\!,\!\mathbf{a}^{(k)}_{t,i}\!;\!\mathbf{z}^{(k)}_{t,i})\! \!\prod\limits_{j=1}^{|\mathcal{M}_{t}^{(k)}|}\!\!h(\mathbf{u}_{t},\breve{\mathbf{v}}^{(k)}_{t,j}\!,\!\mathbf{b}^{(k)}_{t,j}\!;\!\mathbf{z}^{(k)}_{t,j})}_{(b)}\right),
\end{multline}
    \hrulefill
\end{figure*}
where (a), (b), and (c) correspond to the state transition, measurement evaluation, and data association phases, respectively.

\subsubsection{VRP and PA Refinement}\label{refine}
We can obtain the mean and variance of the 
legacy features in the set  $\mathcal{K}_{t-\mathbb{J}}^{(k-1+\mathbb{I})}$
by marginalizing \eqref{final}.
Let $\tilde{{\mathbf{p}}}_{t-\mathbb{J},i}^{(k-1+\mathbb{I})}$ and ${\tilde{\sigma}_{t-\mathbb{J},i}}^{2\ (k-1+\mathbb{I})}$ denote the mean
and variance of the $i$-th legacy feature, respectively, where $i = 1,\ldots,|\mathcal{K}_{t-\mathbb{J}}^{(k-1+\mathbb{I})}|$.
Similarly, let  ${\mathbf{p}}_{t,j}^{(k)}$ and ${\sigma^2}_{t,j}^{(k)}$ denote the mean
and variance of the $j$-th feature obtained at the current time, respectively, where $j = 1,\ldots,|\mathcal{K}_{t}^{(k)}|$.
We define the distance to represent the similarity of the two features as
\begin{equation}\label{dis}
D_{\rm sim}(i,j) = \| {\mathbf{p}}_{t,j}^{(k)} - \tilde{{\mathbf{p}}}_{t-\mathbb{J},i}^{(k-1+\mathbb{I})} \|,
\end{equation}
where $i = 1,\ldots,|\mathcal{K}_{t-\mathbb{J}}^{(k-1+\mathbb{I})}|$ and $j = 1,\ldots,|\mathcal{K}_{t}^{(k)}|$.
We select a threshold $\delta_{\rm sim}$. If $D_{\rm sim}(i,j) < \delta_{\rm sim}$, the two features are regarded as the same feature. 
For the same feature ${\mathbf{p}}_{t,j}^{(k)}$ and $\tilde{{\mathbf{p}}}_{t-\mathbb{J},i}^{(k-1+\mathbb{I})}$, if ${\sigma^2}_{t,j}^{(k)} > {\tilde{\sigma}_{t-\mathbb{J},i}}^{2\ (k-1+\mathbb{I})}$, which means the accuracy of the estimated feature at the current time is worse than that at the previous time, we do not update the feature by letting $\tilde{{\mathbf{p}}}_{t-\mathbb{J},i}^{(k-1+\mathbb{I})}$ replace ${\mathbf{p}}_{t,j}^{(k)}$ in the set $\mathcal{K}_{t}^{(k)}$. Otherwise, we refine the feature with its latest estimate.
The refinement method of PAs is similar to that of VRPs.

\subsubsection{Factor Graph and Message Passing}
A factor graph representing the factorization in \eqref{final} is depicted in Fig. \ref{fig:fg}.
The messages are sent forward in time.
Specifically,
the messages first undergo the state transition phase, in which the messages from  PA $K$ at time $t-1$ pass through factor nodes $f_{\mathbf{u}}$ and $\tilde{f}^{(K)}_{l}$ for $l = 1,\ldots, |\mathcal{K}_{t-1}^{(K)}|$ to generate 
prediction messages.
Second, measurement evaluation calculations are processed by factor node $g^{(1)}_{l}$ for legacy features and factor node $h^{(1)}_{l'}$ for new features in parallel, where $l' = 1,\ldots, |\mathcal{M}_{t}^{(1)}|$.
Third, the output messages of factor nodes $g^{(1)}_{l}$ and $h^{(1)}_{l'}$ are passed to the data association variable nodes $a_{t,l}^{(1)}$ and $b_{t,l'}^{(1)}$, respectively. 
The messages are calculated iteratively among $a_{t,l}^{(1)}$ and $b_{t,l'}^{(1)}$, and this process is called the loopy data association phase.
After the last iteration, the messages are passed back from $a_{t,l}^{(1)}$ to $g^{(1)}_{l}$ and from $b_{t,l'}^{(1)}$ to $h^{(1)}_{l'}$. 
The messages are subsequently updated by factor nodes $g^{(1)}_{l}$ and $h^{(1)}_{l'}$.
Fourth, the updated messages pass through factor nodes $\tilde{f}^{(1)}_{l}$ for $l = 1,\ldots, |\mathcal{K}_{t}^{(1)}|$ to generate 
prediction messages.
The subsequent message passing order is the same as that for PA $1$.
When the message passing is completed in PA $K$ at time $t$, the message passing process at time $t$ is completed.
Lastly, the messages are fused at variable nodes $\mathbf{u}_t$, $\tilde{\mathbf{{v}}}_{t,l}^{(K)}$, and $\breve{\mathbf{{v}}}_{t,l'}^{(K)}$.
Once the messages are available, the belief approximating the desired marginal posterior PDFs is obtained.

\section{Numerical Results}	\label{result}

\begin{figure}
	\centering
	\includegraphics[scale=0.4]{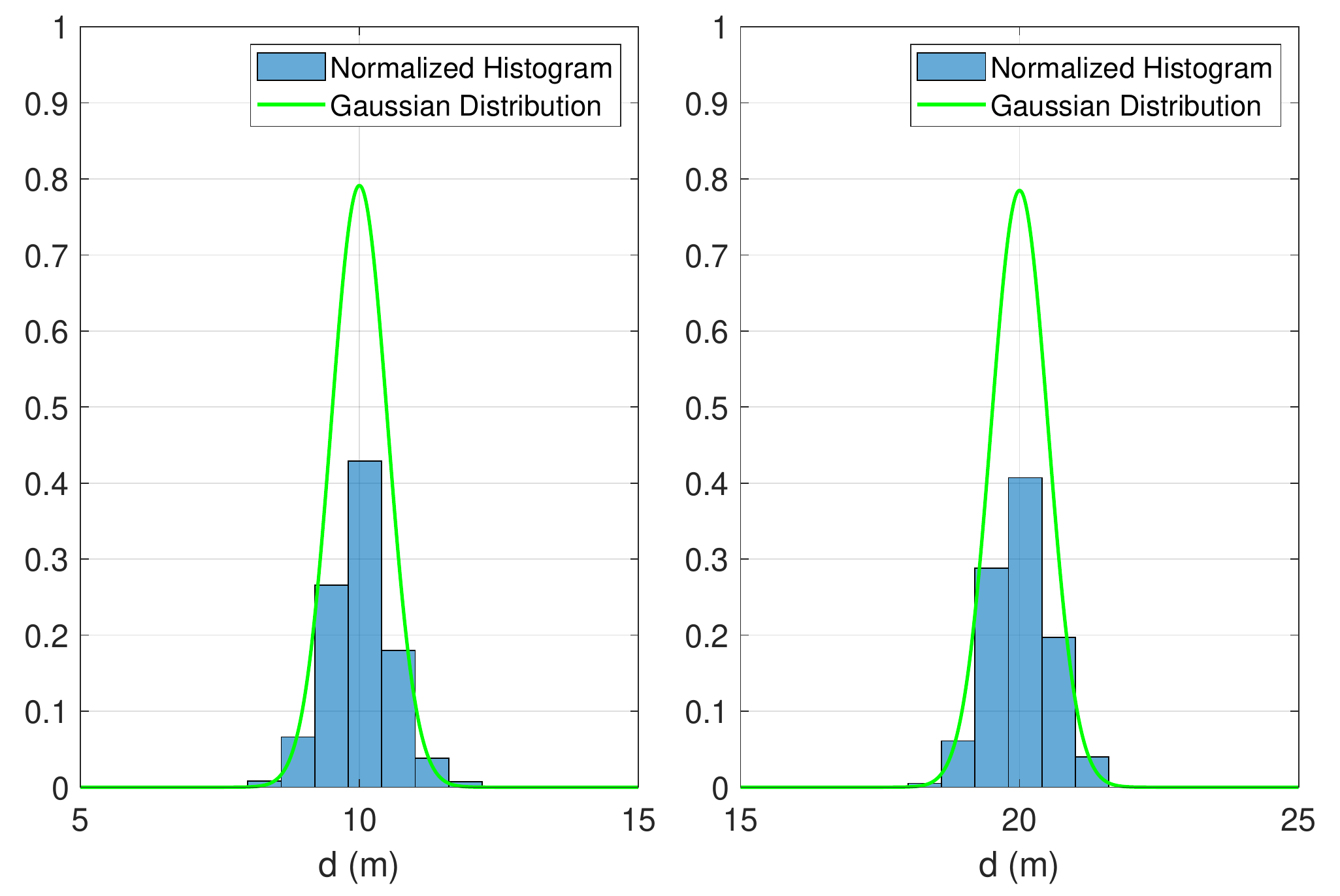}
	\caption{Distribution of estimated distance by the NOMP algorithm.} \label{fig:ddh}	
\end{figure}

\begin{figure}
	\centering
	\includegraphics[scale=0.6]{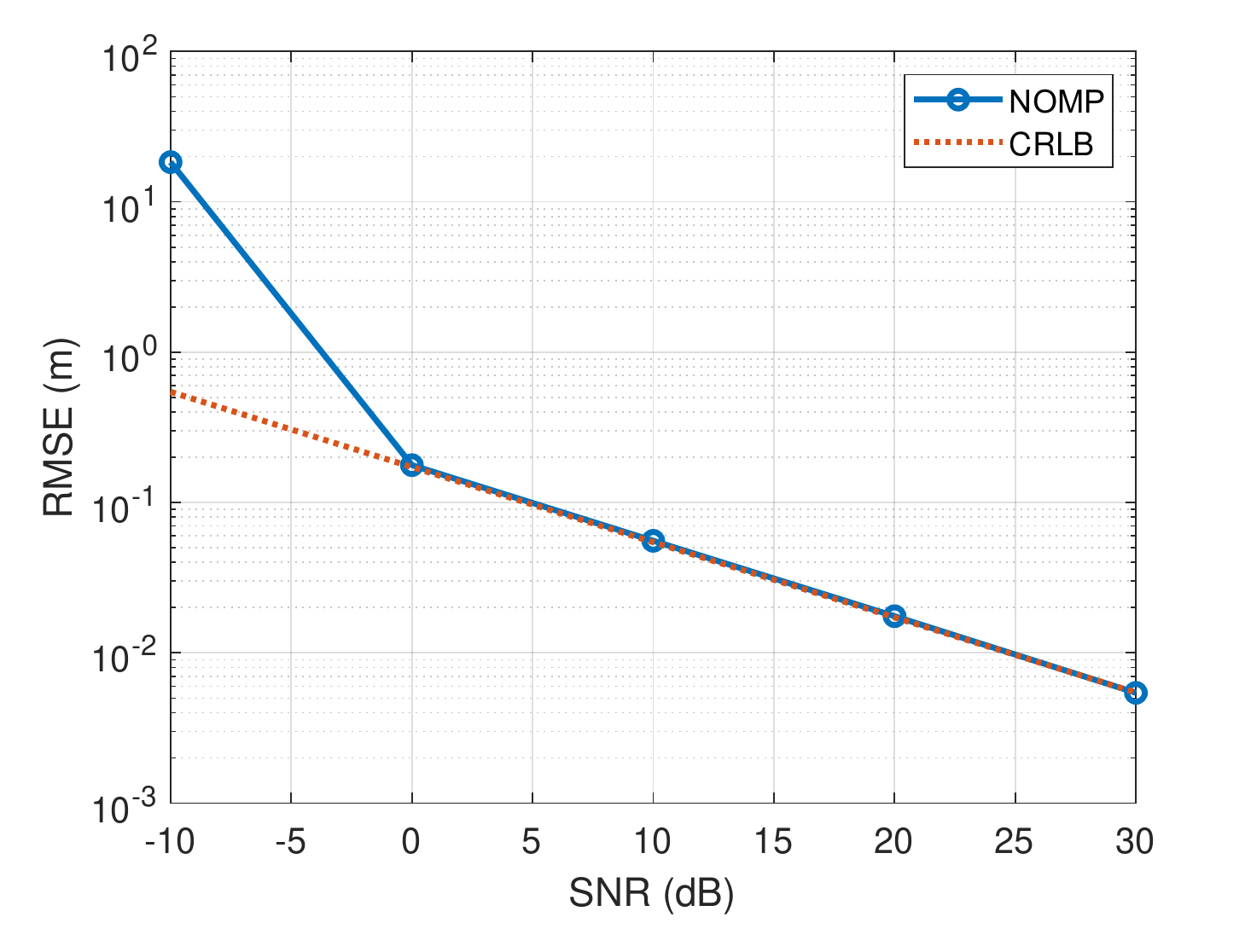}
	\caption{Comparison of RMSE and CRLB of the distance.} \label{fig:nomp}	
\end{figure}

\subsection{Gaussian Distribution Assumption of Distance}
To evaluate the assumption of Gaussian distribution for $d$ in Remark 1, we apply the Newtonized orthogonal matching pursuit (NOMP) algorithm \cite{hy} to extract $d$ from the received signal in \eqref{ren}.
We then plot the normalized histogram of the estimated $d$. 
Finally, we compare the distribution of the estimated $d$ with the Gaussian distribution assumption given in Remark 1. 

The simulation results are shown in Fig. \ref{fig:ddh}, where SNR $=-10$ dB,  the number of subcarrier $N_s = 200$, subcarrier spacing ${\Delta f = 120}$ KHz, RCS ${\varepsilon = 1}$, and the central frequency of carrier ${f=28}$ GHz. 
Distance $d$ is set to $10$ and $20$ m in the left and right hand sides of Fig. \ref{fig:ddh}, respectively. 
We estimate $d$ by using $10,000$ independent Monte Carlo simulations. Thus, we obtain $10,000$ estimations of $d$.
Fig. \ref{fig:ddh} shows the normalized histogram of $10,000$ estimations of $d$ and the Gaussian probability distribution function $\mathcal{N}(d,\sigma^2_d)$, where $\sigma^2_d$ is calculated in accordance with \eqref{variance}.
The results indicate that the estimated $d$ is close to the Gaussian distribution. Therefore, we can assume that $d$ follows the Gaussian distribution.

\subsection{Estimation Performance of Distance}
In this subsection, we evaluate the estimation performance of channel parameters used in the proposed SLAM mechanism.
We take the active sensing as an example, and apply the NOMP algorithm \cite{hy} to estimate the parameter $d$, here we ignore the subscript $m$. 
The simulation result is shown in Fig. \ref{fig:nomp}, where subcarrier spacing ${\Delta f = 120}$ KHz, RCS ${\varepsilon = 1}$, the central frequency of carrier ${f=28}$ GHz, the number of subcarrier $N_s = 200$, and $d=20$ m.
We estimate $d$ by NOMP algorithm for $T=10,000$ independent Monte Carlo simulations, and the estimations are denoted as $\hat{d}_i$, for $i = 1,\ldots,T$. 
The RMSE is defined as 
$\mbox{RMSE}({d})=\sqrt{\sum_{i=1}^{T}||\hat{d}_i-d||^2/T}$.
The CRLB is calculated in accordance with  \eqref{variance}.
The results in Fig. \ref{fig:nomp} show that by using NOMP algorithm, when SNR is larger than $0$ dB, the RMSE can achieve the CRLB, therefore, we can use the derived result in \eqref{variance} as the variance of the estimated $d$.
For passive sensing, the NOMP algorithm can be easily extended to handle the cases with multi-path, and the processes are explained in detail in \cite{hy}.

\subsection{Approximation Error of Distance Uncertainty}
In this subsection, we analyze the approximation error of the distance uncertainty given in Theorem \ref{T1}.
Given that we consider a single distance in this section, we ignore the subscript $m$.
In the following simulation, subcarrier spacing ${\Delta f = 120}$ KHz, RCS ${\varepsilon = 1}$, and the central frequency of carrier ${f=28}$ GHz. 
We denote the lower bound of the variance calculated according to \eqref{fim} as  ``var$_{\rm true}$", and calculated according to \eqref{variance} as ``var$_{\rm appr}$".

In Fig.\ref{fig:variance} (a), 
the number of subcarriers is ${N_{\rm s}=1,000}$, thus, the bandwidth ${B=120}$ MHz.
When $d$, the distance between the RSP to the RP, increases, the variance increases.
In Fig.\ref{fig:variance} (b), distance $d$ is fixed to $10$ m, and SNR is fixed to $0$ dB. When the number of subcarriers $N_{\rm s}$ increases from $200$ to ${1,000}$, the variance decreases.
The approximation error is negligible when 
$N_{\rm s}$ is larger than $400$.
In Fig.\ref{fig:variance} (c), distance $d$ is fixed to $10$ m, and the number of subcarriers $N_{\rm s}$ is set to ${1,000}$.
The variance decreases when SNR increases.
According to the simulation result, the approximation error is negligible. Thus, the derived $ {\big( \frac{8\pi^2B^2N_{\rm s}}{3c^2} {\rm SNR}_{m} \big)^{-1}}$ as the lower bound of the variance of $\hat{d}_m$ is reasonable, as described in Theorem \ref{T1}.

\subsection{Approximation Error of VRP Uncertainty}
In this subsection, we investigate the approximation error of VRP uncertainty, because the conclusions from Theorem \ref{T2} and Corollary \ref{L2} are obtained by first-order Taylor approximations.
In the following simulation, 
 subcarrier spacing ${\Delta f = 120}$ KHz,
the number of subcarriers is ${N_{\rm s}=200}$,
the central frequency of carrier ${f=28}$ GHz,
RP ${=[0,0]}$,
SNR is $0$ dB when the smallest distance from RP to the reflective surface is $10$ m, and $3$ RSPs are considered.
When modeling RCS by ${\varepsilon_m = \gamma \cos^{2\eta} \psi_m}$, we select ${\gamma=1}$ and ${\eta=0.2}$.

The least squares (LS) estimator is adopted as the benchmark.
We estimate VRP by using the LS estimator for ${10,000}$ independent Monte Carlo simulations and then calculate the mean and variance of the ${10,000}$ results.
Next, we compare the mean and variance of the VRP calculated by Theorem \ref{T2} and Corollary \ref{L2} with those calculated by the LS estimator.
``LS" denotes the result of the benchmark, and ``Taylor Appr." represents the result of the proposed method.
Figs. \ref{fig:mvavariance} (a) and (b) show the Gaussian distribution of the x and y coordinates of VRP, respectively, 
when the smallest distance between RP and the reflective surface is $50$ m.
The upper and lower subfigures in
Figs. \ref{fig:mvavariance} (c) and (d) show the Gaussian distribution of the x and y coordinates of VRP, respectively,  
when the smallest distance between RP and the reflective surface is ${1,000}$ m.
The approximation error of the mean and variance obtained by ``LS" and ``Taylor Appr." is acceptable.
Therefore, the conclusions from Theorem \ref{T2} and Corollary \ref{L2} can be used in most scenarios with a range below the kilometer level.

\begin{figure}
	\centering
	\includegraphics[scale=0.5]{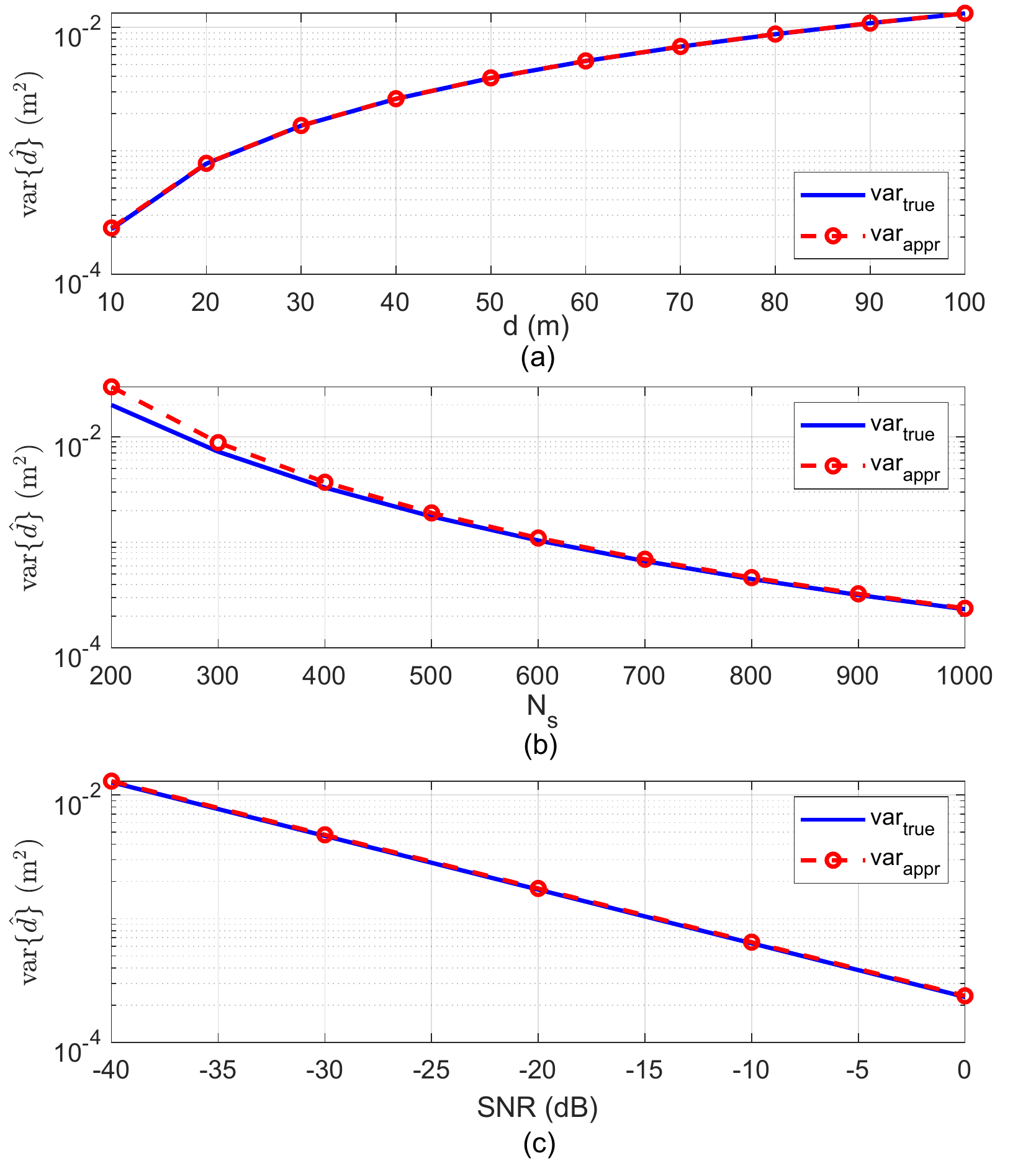}
	\caption{Comparison of the approximated lower bound of variance with the true value.} \label{fig:variance}	
\end{figure}

\begin{figure}
	\centering
	\vspace{-0.15cm}
	\includegraphics[scale=0.41]{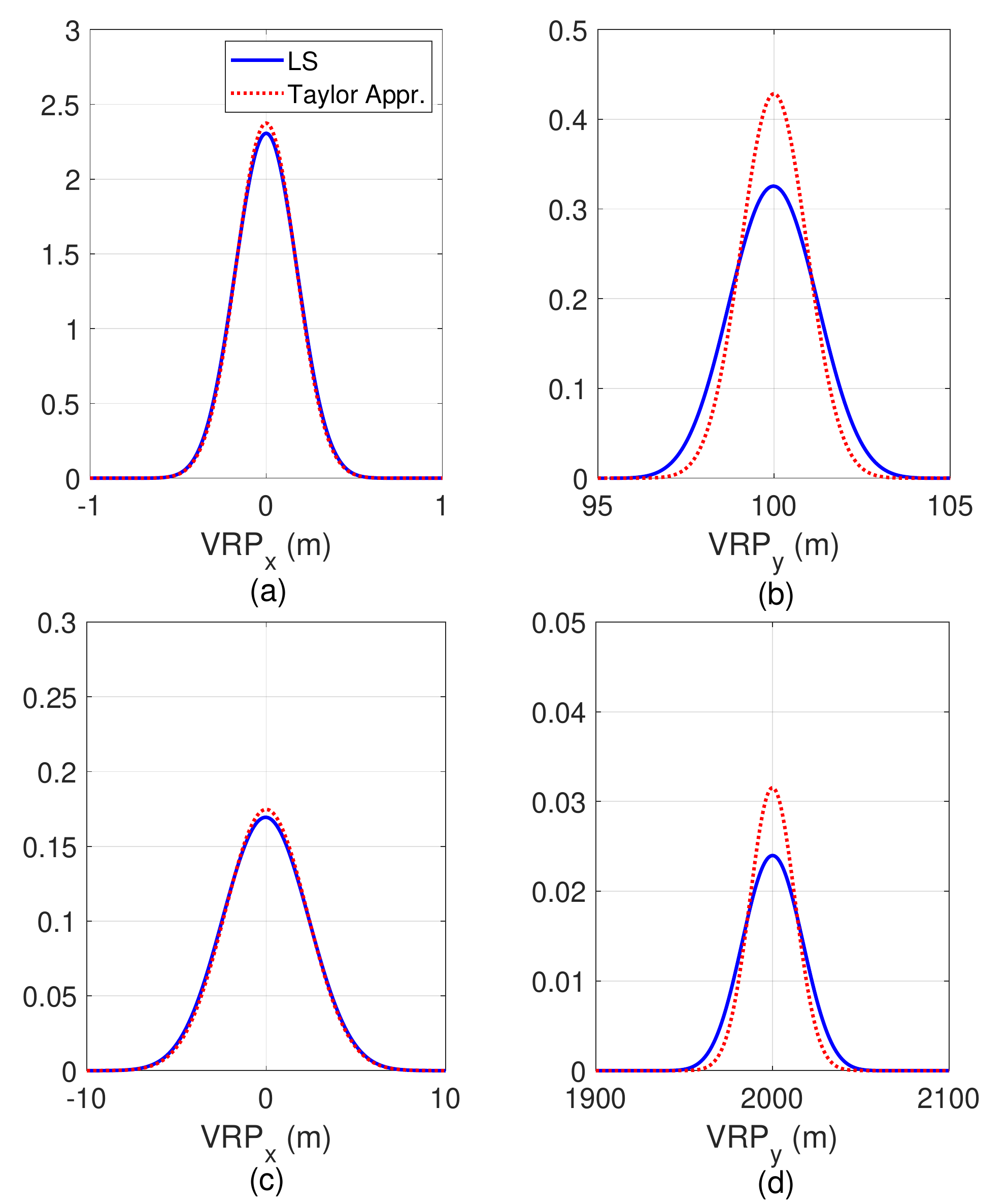}
	\caption{ Comparison of the VRP distribution obtained by Taylor approximation with that from LS-based Monte-Carlo simulation. (a) and (b) denote the Gaussian distribution of the x and y coordinates of VRP, respectively, 
		when the smallest distance between  RP and the reflective surface is $50$ m.
		(c) and (d) show the Gaussian distribution of the x and y coordinates of VRP, respectively, 
		when the smallest distance between  RP and the reflective surface is ${1,000}$ m.
	} \label{fig:mvavariance}	
\end{figure}

\begin{figure*}
	\centering
	\includegraphics[scale=0.64]{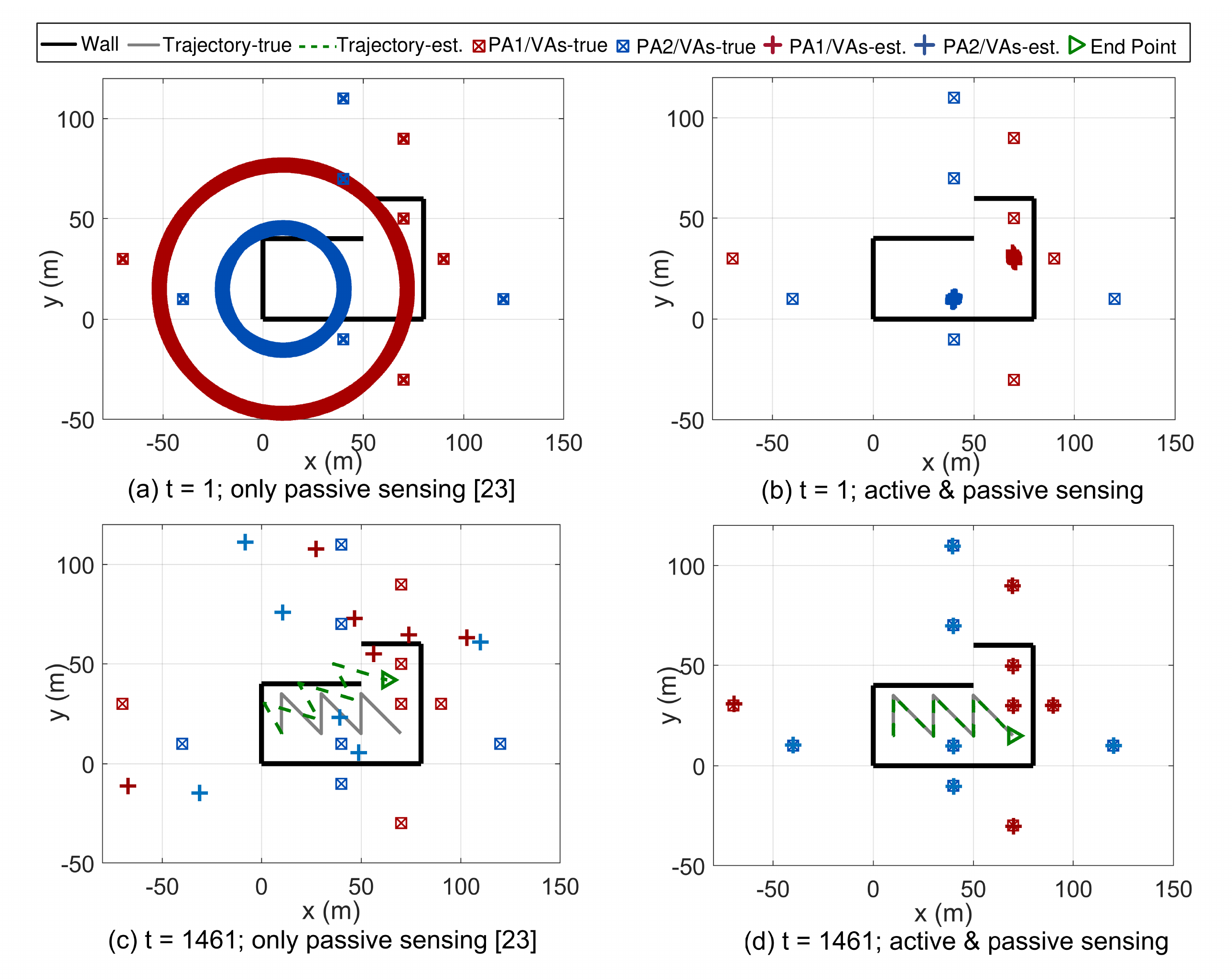}
	\caption{Comparison of the proposed mechanism with only the passive sensing mechanism \cite{BP2}. (a) and (b) depict the estimated PA1 and PA2 when ${t = 1}$. 
		(c) and (d) depict the estimated PAs, VAs, and agent trajectory when ${t = 1,461}$. 	
	} \label{fig:initialPA}	
\end{figure*}

\subsection{Performance of the Proposed SLAM Mechanism}

In this subsection, we analyze the performance of the proposed SLAM mechanism that combines
active and passive sensing.
We use the floor plan shown in Fig. \ref{fig:initialPA}, the size of which is approximately $80\times 60\  \rm{m}^2$.
The ROI is a circular disk with a radius of $180$ m.
Most of the messages leaving the factor nodes cannot be solved in a closed form due to the contained integrals. 
We use particle-based implementation to approximate the continuous messages \cite{BP1}.
The measurement noise follows a Gaussian distribution with zero mean and standard deviation of $\sigma_t=0.1$ m for TOA measurements.
The state transition of the agent is given by $\mathbf{u}_{t}^{\rm T} = \mathbf{A} \mathbf{u}^{\rm T}_{t-1} + \bm{\omega}_{t}$, where the variance of the driving process $\bm{\omega}_{t}$ is $0.0278$.
The state transition PDFs of features are given by Dirac delta functions. 
In accordance with \cite{BP2}, we introduced a small driving process for numerical stability, and the variance of the driving process is $ 10 ^ {-8}$.
The parameters involved in the algorithm are as follows:  detection probability 
$P_{\rm d}=0.95$;  survival probability $P_{\rm s}=0.999$;  mean of false alarms $\mu_{\rm false}=1$;  mean of newly born features $\mu_{\rm new}=10^{-4}$;  step length is $0.1$ m;  unreliability threshold is $ 10^{-4}$;  similarity threshold  $\delta_{\rm sim}$ is $1$ m;  detection threshold is $0.5$; and the number of particles is $10^5$.

\subsubsection{Assistance of Active Sensing to Passive Sensing}\label{N1}

In this subsection, we analyze the assistance provided by active sensing to passive sensing in the proposed mechanism. We initialize the location of PAs with the assistance of the prior VRPs obtained by active sensing.
The start point of the mobile agent is $[10,15]$, which is also the RP. 
The mean and variance of the VRPs can be calculated in accordance with Corollary \ref{L2}.
When subcarrier spacing ${\Delta f = 120}$ KHz,
the number of subcarriers is ${N_{\rm s}=200}$, the central frequency of carrier ${f=28}$ GHz, SNR is $0$ dB when the smallest distance from RP to the reflective surface is $10$ m, and two RSPs are considered.
According to the floor plan, five walls (reflective surfaces) are present, but only four are observed by active sensing when ${t = 1}$.
The means of four VRPs are $[-10, 15]$, $[10,65]$, $[150, 15]$, and $[10, -15]$,
and the variances of the four VRPs obtained in active sensing are $0.08$, $0.7$, $4$, and $0.2$ m$^2$.

We compare the proposed mechanism in which active sensing assists passive sensing with the passive sensing only mechanism \cite{BP2}. The results are shown in Fig. \ref{fig:initialPA}. 
Without the assistance of active sensing, the PAs are distributed on the circles, as shown in Fig. \ref{fig:initialPA} (a).
With the assistance of active sensing, 
the initialized PAs converge near the true locations when $t=1$, as shown in Fig. \ref{fig:initialPA} (b).
With the accumulation over time, the PAs are refined, and the corresponding VAs and the trajectory are obtained accurately by the proposed mechanism when ${t=1,460}$, as shown in Fig. \ref{fig:initialPA} (d).
However, the estimated VAs and trajectory rotate when the passive sensing only mechanism is used, as shown in Fig. \ref{fig:initialPA} (c),
because the possible values of PA are distributed on the circle without any prior information about the PA, as shown in Fig. \ref{fig:initialPA} (a). When PA converges to a wrong point on the circle, BP SLAM generates features and the trajectory relative to the wrong PA, leading to the rotation of the result.

\begin{figure*}
	\centering
	\includegraphics[scale=0.64]{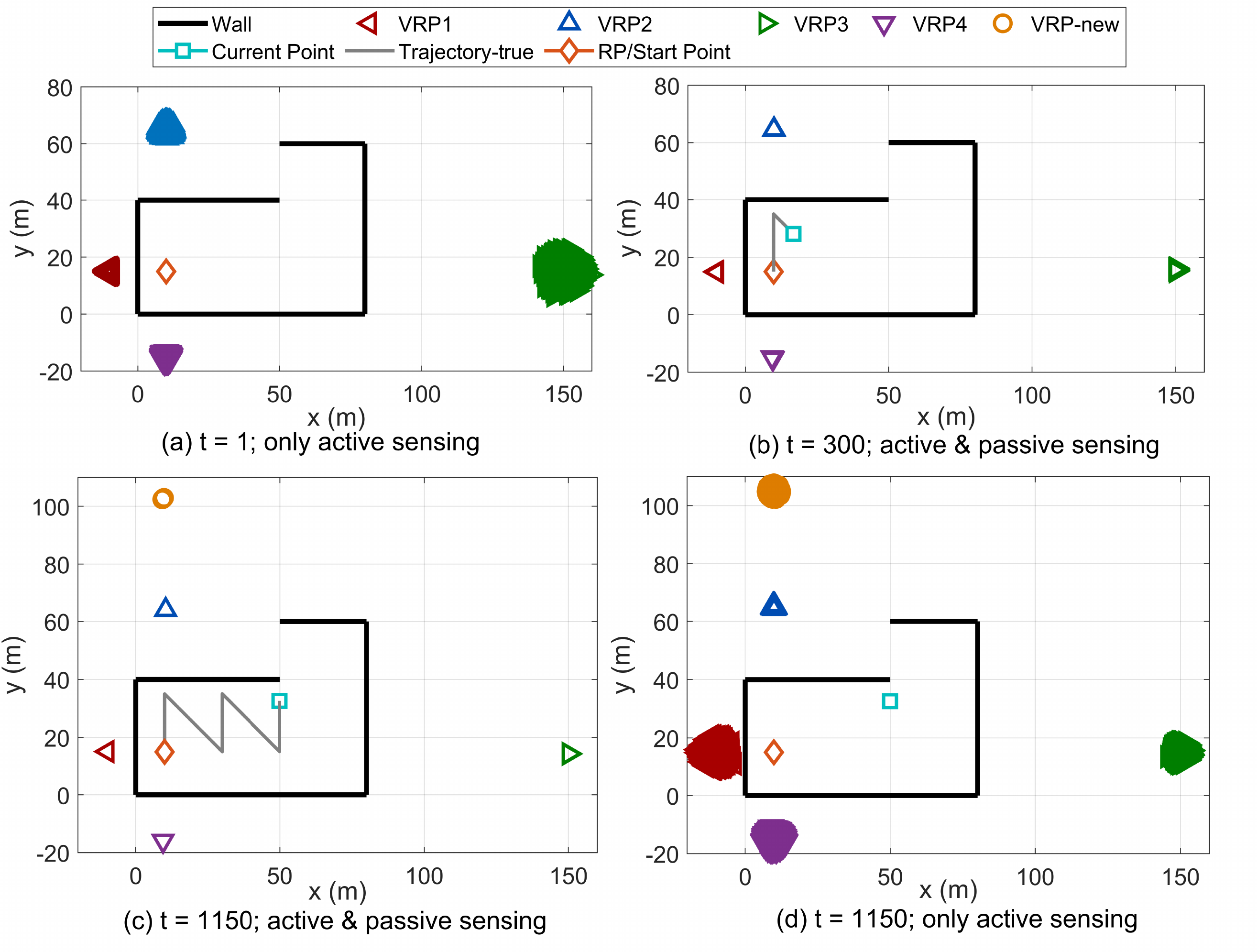}
	\caption{Comparison of the VRP estimation performance of the proposed mechanism and that of the active sensing only mechanism. (a) depicts the result of active sensing when ${t=1}$, and the result is used as the prior information of the VRPs. (b) plots the result of the proposed mechanism when ${t=300}$, and the accuracy of the estimated VRPs are improved. (c) shows the result of the proposed mechanism when ${t=1,150}$, and the accuracy of the estimated VRPs are further improved. Moreover, a new VRP is detected. (d) depicts the result of the active sensing only mechanism when the agent is located at ${t=1,150}$. 
	} \label{fig:mvarefine}	
\end{figure*}

\subsubsection{Assistance of Passive Sensing to Active Sensing}\label{ex}

In this subsection, we investigate the assistance provided by passive sensing to active sensing.
The simulation settings are similar to those in Section \ref{N1}.
When ${t= 1}$, four walls are observed by active sensing, but the wall that starts from $[50,60]$ and ends at $[80,60]$ is blocked.
Therefore, we obtain four VRPs with
mean values of $[-10, 15]$, $[10,65]$, $[150, 15]$, and $[10, -15]$ and variances of $0.08$, $0.7$, $4$, and $0.2$ m$^2$.
The particles of VRPs obtained by active sensing when $t=1$ are depicted in Fig. \ref{fig:mvarefine} (a).
The farther the wall is from RP, the greater the variance of the corresponding VRP is.
``VRP$m$-prior" denotes the particles of VRP$m$ generated by active sensing when ${t=1}$.
The prior information of VRPs is used to initialize the PAs.
Then, the proposed mechanism can work.
As shown in Fig. \ref{fig:mvarefine} (b),
the variances of four VRPs are reduced when ${t=300}$,
which verifies the capability of passive sensing in VRP refinement with the proposed mechanism.
Moreover, when ${t=1,150}$, the wall that starts from $[50,60]$ and ends at $[80,60]$ is detected by passive sensing with the proposed mechanism, as shown in Fig. \ref{fig:mvarefine} (c), where five VRPs are estimated and refined.
For comparison, we depict the VRPs obtained by the active sensing when ${t=1,150}$ in Fig. \ref{fig:mvarefine} (d).
Although five VRPs are detected by active sensing, the variances are much larger than that those obtained by the proposed mechanism.
Therefore, with the enrichment of the trajectory, passive sensing can extend the capability of active sensing.
Legacy VRPs can be refined, and new VRPs can be estimated accurately by passive sensing without any prior information.

\begin{figure*}
	\centering
	\includegraphics[scale=0.85]{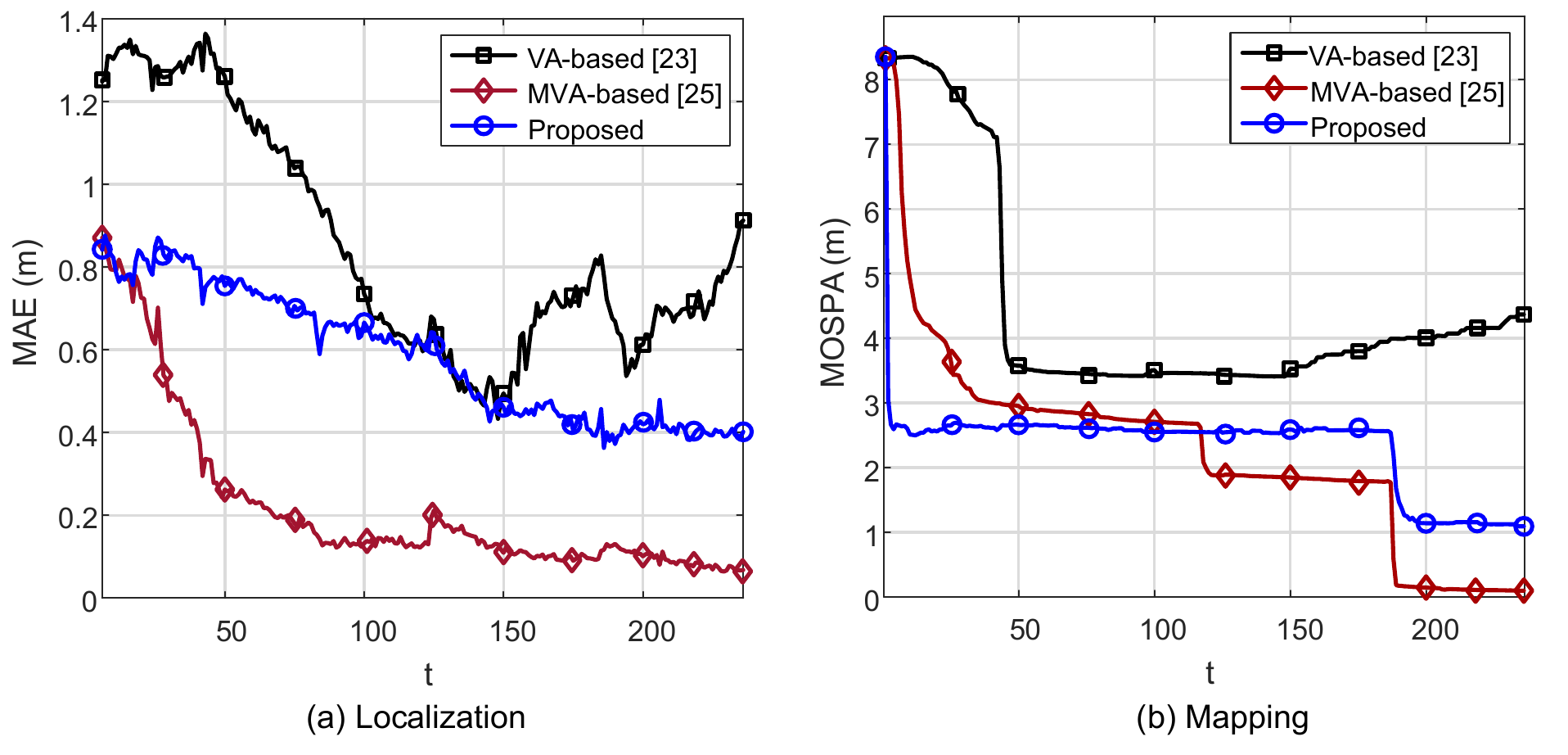}
	\caption{Comparison of the localization and mapping performance obtained by $100$ times of Monte Carlo simulations.} \label{fig:benchmark}	
\end{figure*}

\begin{figure*}
	\centering
	\includegraphics[scale=0.658]{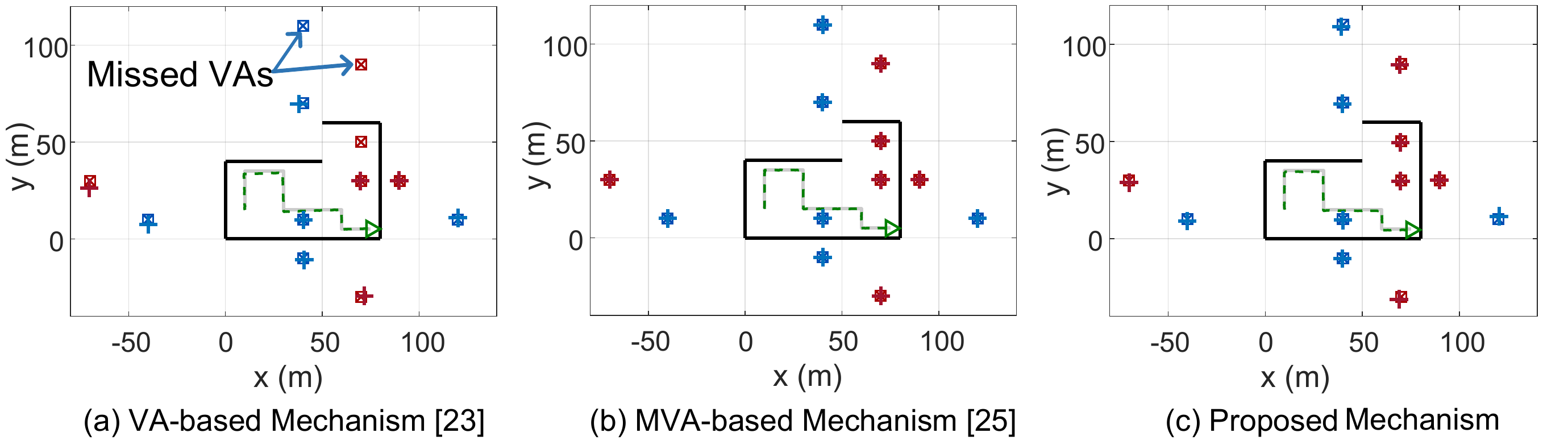}
	\caption{SLAM performance for a single simulation.} \label{fig:slambenchmark}	
\end{figure*}
\subsubsection{Comparison of Different SLAM Mechanisms}

In this subsection, we compare the proposed mechanism with the VA-based \cite{BP2} and  MVA-based \cite{MVA} mechanisms.
The locations of two PAs are perfectly known in the VA and  MVA-based mechanisms.
However, no prior information of PAs is available in 
the proposed mechanism.
The step length is set to $0.5$ m, the true trajectory is shown as the grey solid line in Fig. \ref{fig:slambenchmark}, and the other simulation settings are similar to those in Section \ref{N1}.

The results in Fig. \ref{fig:benchmark} are obtained by $100$ times of Monte Carlo simulations.
Therefore, we select the mean absolute error (MAE) and mean optimal subpattern assignment (MOSPA) error to measure the performance of localization and mapping, respectively.
As shown in Fig. \ref{fig:benchmark} (a), the MVA-based mechanism exhibits the best performance in localization, whereas the VA-based mechanism has the worst performance, because the MVA-based mechanism knows the perfect locations of PAs but the proposed mechanism does not.
Although the VA-based mechanism also knows the perfect locations of PAs, several VAs may be missed over time because only two PAs are available, as shown in Fig. \ref{fig:slambenchmark} (a).
When ${t=236}$,
the localization error of the MVA-based mechanism is around $0.05$ m, and that of the proposed mechanism is around $0.4$ m.
The SLAM result for a single simulation is shown in Figs. \ref{fig:slambenchmark} (b) and (c) for MVA-based and proposed mechanisms, respectively.

With regard to mapping, as shown in Fig. \ref{fig:benchmark} (b), 
the proposed mechanism has the highest convergence speed of mapping (when ${t = 3}$). The MVA-based mechanism converges when ${t=30}$, and the VA-based mechanism has the lowest convergence speed (when ${t = 45}$).
The distributions of PAs and VRPs are relatively accurate with the help of active sensing when ${t=1}$ in the proposed mechanism,
whereas the MVA-based mechanism has no prior information about VRPs (or MVAs).
As mentioned in Section \ref{gm}, the number of VAs is larger than that of VRPs (or MVAs).
Given that MVAs are not introduced in the VA-based mechanism, VA-based mechanism has a slower convergence speed because it has more features than the MVA-based mechanism.
After ${t=120}$, 
the MOSPA error of the MVA-based mechanism is reduced corresponding to the turning of the trajectory.
After ${t=180}$, the fifth wall is detected, similar to the explanation in Section \ref{ex}.
Therefore, the MOSPA error of the MVA-based mechanism is around $0.1$ m after ${t=180}$, and that of the proposed mechanism is around $1$ m.
The difference in the performance of the MVA-based and proposed mechanisms is due to the accuracy of the PAs.
Given that several VAs are missed over time in VA-based mechanism, the MOSPA error gradually increases after ${t=150}$.
Although the proposed mechanism has a certain performance loss compared with the MVA-based mechanism, it does not need any prior information about the PAs, which expands the application scenarios of the proposed mechanism. 

Finally, we summarize the distinct performance of active and passive sensing to make the conclusion clear. 
Active sensing can be performed without establishing a connection in the communication network, thus, it plays an important role in the initial stage.
However, a finite number of discrete beams are used in the beam sweeping phase. Therefore, active sensing can only obtain the information from fixed beam directions.	
On the contrary, passive sensing requires the establishment of communication links. 
Passive sensing can obtain more measurements with better quality than active sensing, because channel estimation is conducted more frequently than beam sweeping in accordance with 5G NR standards, and more pilot resources are allocated to the channel estimation phase.
However, clock and orientation bias are involved in practice.
The proposed hybrid mechanism gives full play to the respective advantages of active and passive sensing.

\section{Conclusion}

In this study, we integrated active and passive sensing for SLAM in wireless communication systems.
Specifically, active sensing was realized by the beam sweeping, and passive sensing was implemented by DL-PRS.
We adopted the idea of VRP, which characterizes the state of the reflective surface.
Therefore, the results of active and passive sensing can be transformed into VRPs.
An uncertainty model of the VRP obtained by active sensing was established to provide the mean and variance of the estimated VRP for soft information fusion with passive sensing.
Next, we extended the classic BP SLAM mechanism by realizing PA initialization with the assistance of active sensing, and achieving VRP and PA refinement with the help of passive sensing.
The numerical results showed that the proposed mechanism works successfully in realistic scenarios without any prior information about the floor plan, anchors, or agents.
Compared with active or passive sensing only mechanisms, the proposed mechanism can mutually enhance the two sensing modes and bring a significant performance gain. For future work, more types of measurement, including angle and Doppler, and unknown measurement bias, such as clock and orientation bias, should be considered to extend the proposed SLAM mechanism. Hybrid precoding and multi-beam sweeping can be studied for active sensing.  
Different materials of reflective surface and multi-user scenarios can also be considered.
We can establish prototype verification systems and collect experimental data to verify the proposed algorithms.

\begin{appendices}
		
\section{Proof of Theorem \ref{T1}}\label{A} 
In this section, we derive the distance uncertainty.
Since we have
\begin{equation}\label{partiald}
\dfrac{\partial {r}_{n,m} }{\partial d_m} \!=\! \frac{b_{n,m}}{d_m^2}e^{-j2\pi  (n-\frac{N_{\rm s}}{2})  \Delta f \frac{2d_m}{c}} \!\!\bigg[-\frac{2}{d_m}-j2\pi \big(n-\frac{N_{\rm s}}{2}\big) \Delta f \frac{2}{c}\bigg],
\end{equation}
and
\begin{equation}\label{partiald2}
\dfrac{\partial {r}_{n,m}^* }{\partial d_m} = \frac{b_{n,m}}{d_m^2}e^{j2\pi (n-\frac{N_{\rm s}}{2}) \Delta f \frac{2d_m}{c}} \bigg[-\frac{2}{d_m}+j2\pi \big(n-\frac{N_{\rm s}}{2}\big) \Delta f \frac{2}{c}\bigg].
\end{equation}
Then, we obtain
\begin{equation}\label{partiald3}
 \dfrac{\partial {r}_{n,m}^* }{\partial d_m}\dfrac{\partial {r}_{n,m} }{\partial d_m} = \frac{b_{n,m}^2}{d_m^4} \bigg[\frac{4}{d_m^2}+4\pi^2 \big(n-\frac{N_{\rm s}}{2}\big)^2 \Delta f^2 \frac{4}{c^2}\bigg].
\end{equation}
Therefore, we get
\begin{multline}\label{fimfinal}
{\rm F}(d_m) = \frac{2}{\sigma^2}\mathcal{R} \bigg\{\dfrac{\partial \textbf{r}_{m}^{\text{H}} }{\partial d_m}  \dfrac{\partial \textbf{r}_{m} }{\partial d_m} \bigg\}\\
=\frac{2}{\sigma^2}\mathcal{R} \bigg\{ \sum_{n=1}^{N_{\rm s}}\frac{b_{n,m}^2}{d_m^4} \bigg[\frac{4}{d_m^2}+4\pi^2 \big(n-\frac{N_{\rm s}}{2}\big)^2 \Delta f^2 \frac{4}{c^2}\bigg]\bigg\}.
\end{multline}
Given $\sum_{n=1}^{N_{\rm s}}\big(n-\frac{N_{\rm s}}{2}\big)^2=\dfrac{ N_{\rm s}^3+2N_{\rm s}}{12}$, and we assume that $\lambda_n\approx \lambda$ for mmWave frequencies, therefore, we have $\lvert b_{n,m} \rvert^2 = \frac{\lambda_n^2 \varepsilon_m}{(4\pi)^3} \approx \frac{\lambda^2 \varepsilon_m}{(4\pi)^3}$,
after some derivations, we get
\begin{equation}\label{fimappr}
{\rm F}(d_m) = \frac{2}{\sigma^2} \bigg[
\frac{\lambda^2\varepsilon_m \Delta f^2 N_{\rm s}^3}{48\pi c^2 d_m^4}   +\frac{\lambda^2\varepsilon_m\Delta f^2 N_{\rm s}}{24\pi c^2 d_m^4} + \frac{\lambda^2\varepsilon_m N_{\rm s}}{16\pi^3 d_m^6} \bigg].
\end{equation}    
Then, we have 
\begin{equation}\label{fimappr1}
{\rm F}(d_m) = \frac{\lambda^2\varepsilon_m \Delta f^2N_{\rm s}^3}{24\pi\sigma^2c^2 d_m^4} \bigg[ 1 + \frac{2}{N_{\rm s}^2} + \frac{3 c^2 }{\pi^2 d_m^2N_{\rm s}^2 \Delta f^2} \bigg].
\end{equation} 
When $N_{\rm s} \gg 1$, we have 
\begin{equation}\label{fimappr2}
{\rm F}(d_m) = \frac{\lambda^2\varepsilon_m \Delta f^2 N_{\rm s}^3 }{24\pi\sigma^2c^2 d_m^4} \bigg[ 1 + O\left(\dfrac{1}{N_{\rm s}}\right) \bigg].
\end{equation} 
Therefore, we assume that 
\begin{equation}\label{fimappro}
{\rm F}(d_m) \approx 
\frac{\lambda^2\varepsilon_m \Delta f^2 N_{\rm s}^3}{24\pi\sigma^2c^2 d_m^4}.
\end{equation}
Define ${\rm SNR}_{m} = \dfrac{\lvert  b_{n,m} \rvert^2}{d_m^4 \sigma^2} \approx \dfrac{ \lambda^2 \varepsilon_m }{(4\pi)^3 d_m^4 \sigma^2}$, and $B = N_{\rm s}\Delta f$, where $B$ is the bandwidth, we have
\begin{equation}\label{variance1}
{\rm var}\{\hat{d}_m\} \geq  \bigg( \dfrac{8\pi^2B^2N_{\rm s}}{3c^2} {\rm SNR}_{m} \bigg)^{-1}.
\end{equation}

\section{Proof of Theorem \ref{T2}}\label{B} 
In this section, we derive the first-order Taylor approximation of the VRP.
For a quaternion function $f(x_1, x_2, y_1, y_2)$, where ${x_1 \rightarrow \mu_{x1}, x_2 \rightarrow \mu_{x2}, y_1 \rightarrow \mu_{y1}, y_2 \rightarrow \mu_{y2}}$, let $\bm{\mu}=(\mu_{x1},\mu_{x2},\mu_{y1},\mu_{y2})$, there holds
\begin{multline}
f(x_1, x_2, y_1, y_2) = f({\bm{\mu}}) + \dfrac{\partial f }{\partial x_1} \bigg|_{\bm{\mu}}\!\!\!(x_1-\mu_{x1}) + \dfrac{\partial f }{\partial x_2} \bigg|_{\bm{\mu}}\!\!\!(x_2-\mu_{x2}) \\
+ \dfrac{\partial f }{\partial y_1} \bigg|_{\bm{\mu}}(y_1-\mu_{y1}) + \dfrac{\partial f }{\partial y_2} \bigg|_{\bm{\mu}}(y_2-\mu_{y2}) + o(\rho),
\end{multline} 
where 
\begin{equation}
\rho \!= \! \! \sqrt{(x_1  \!- \!\mu_{x1})^2+(x_2 \!- \!\mu_{x2})^2+(y_1 \!- \!\mu_{y1})^2+(y_2 \!- \!\mu_{y2})^2},
\end{equation} 
and $\rho \rightarrow 0$.
According to \eqref{x} and \eqref{y}, we have
\begin{equation}\label{xappro}
x \approx \frac{a_0}{b_0}+\mathbf{w}\mathbf{q}_x, \ y \approx \frac{c_0}{b_0}+\mathbf{w}\mathbf{q}_y,
\end{equation}
where  
\begin{equation}
\hspace{-1.82cm}\mathbf{w}=[x_1-\mu_{x1},x_2-\mu_{x2},y_1-\mu_{y1},y_2-\mu_{y2}],
\end{equation}
\begin{equation}
\mathbf{q}_x
\!\!=\!\!\bigg[\dfrac{a_1b_0\!-\!a_0b_1}{b_0^2},\dfrac{a_2b_0\!-\!a_0b_2}{b_0^2},\dfrac{a_3b_0\!-\!a_0b_3}{b_0^2},\dfrac{a_4b_0\!-\!a_0b_4}{b_0^2}\bigg]^{\rm T}\!\!, \end{equation}
and
\begin{equation}  
\mathbf{q}_y\!\!=\!\!\bigg[\dfrac{c_1b_0\!-\!c_0b_1}{b_0^2}, \dfrac{c_2b_0\!-\!c_0b_2}{b_0^2}, \dfrac{c_3b_0\!-\!c_0b_3}{b_0^2},\dfrac{c_4b_0\!-\!c_0b_4}{b_0^2}\bigg]^{\rm T},
\end{equation}
with
\begin{align}
a_0 = & x_{\rm rp}(\mu_{x2}-\mu_{x1})^2 - (x_{\rm rp}-2\mu_{x1})(\mu_{y2}-\mu_{y1})^2 \nonumber \\
& + 2(y_{\rm rp}-\mu_{y1})(\mu_{y2}-\mu_{y1})(\mu_{x2}-\mu_{x1}),\\
a_1 = & -2x_{\rm rp}(\mu_{x2}-\mu_{x1}) + 2(\mu_{y2}-\mu_{y1})^2 \nonumber\\
& - 2(y_{\rm rp}-\mu_{y1})(\mu_{y2}-\mu_{y1}),\\ 
a_2 = & 2x_{\rm rp}(\mu_{x2}-\mu_{x1}) + 2(y_{\rm rp}-\mu_{y1})(\mu_{y2}-\mu_{y1}),\\  
a_3 = & 2(x_{\rm rp}-2\mu_{x1})(\mu_{y2}-\mu_{y1}) \nonumber\\
& +
2(\mu_{x2}-\mu_{x1})(-\mu_{y2}-y_{\rm rp}+2\mu_{y1}),\\  
a_4 = & -2(x_{\rm rp}-2\mu_{x1})(\mu_{y2}-\mu_{y1}) \nonumber\\
& + 2(y_{\rm rp}-\mu_{y1})(\mu_{x2}-\mu_{x1}),\\
b_0 = & (\mu_{x2}-\mu_{x1})^2 + (\mu_{y2}-\mu_{y1})^2,\\
b_1 = & -2(\mu_{x2}-\mu_{x1}), \\
b_2 = & 2(\mu_{x2}-\mu_{x1}), \\   
b_3 = & -2(\mu_{y2}-\mu_{y1}), \\
b_4 = & 2(\mu_{y2}-\mu_{y1}), 
\end{align}
\begin{align}  
c_0 = & y_{\rm rp}(\mu_{y2}-\mu_{y1})^2 - (y_{\rm rp}-2\mu_{y1})(\mu_{x2}-\mu_{x1})^2 \nonumber\\
& + 2(x_{\rm rp}-\mu_{x1})(\mu_{y2}-\mu_{y1})(\mu_{x2}-\mu_{x1}),\\
c_1 = & 2(y_{\rm rp}-2\mu_{y1})(\mu_{x2}-\mu_{x1})\nonumber\\
& + 2(\mu_{y2}-\mu_{y1})(-\mu_{x2}-x_{\rm rp}+2\mu_{x1}),\\  
c_2 = & -2(y_{\rm rp}-2\mu_{y1})(\mu_{x2}-\mu_{x1}) \nonumber\\
& + 2(x_{\rm rp}-\mu_{x1})(\mu_{y2}-\mu_{y1}),\\
c_3 = & -2y_{\rm rp}(\mu_{y2}-\mu_{y1}) + 2(\mu_{x2}-\mu_{x1})^2 \nonumber\\
& - 2(x_{\rm rp}-\mu_{x1})(\mu_{x2}-\mu_{x1}),\\
c_4 = & 2y_{\rm rp}(\mu_{y2}-\mu_{y1}) + 2(x_{\rm rp}-\mu_{x1})(\mu_{x2}-\mu_{x1}).
\end{align}

\end{appendices}

\bibliographystyle{IEEEtran}
\bibliography{bibsample}

\end{document}